\documentclass[letterpaper]{article} 
\usepackage{aaai25}  
\usepackage{times}  
\usepackage{helvet}  
\usepackage{courier}  
\usepackage[hyphens]{url}  
\usepackage{graphicx} 
\usepackage{natbib}  
\usepackage{caption} 
\usepackage{algorithm}
\usepackage{algorithmic}

\usepackage{newfloat}
\usepackage{listings}
\DeclareCaptionStyle{ruled}{labelfont=normalfont,labelsep=colon,strut=off} 
\floatstyle{ruled}
\newfloat{listing}{tb}{lst}{}
\floatname{listing}{Listing}
\usepackage{amsmath}
\usepackage{amsthm}
\usepackage{xcolor}
\usepackage{amsfonts}
\usepackage{subcaption}
 
\usepackage{pifont}
\usepackage{latexsym}
\usepackage{multirow}
\usepackage{tikz}
\usepackage{tipa}
\usepackage{dsfont}
\usepackage{color} 
\usepackage{array}
\usepackage{booktabs}
\usepackage{graphicx}
\newtheorem{definition}{Definition}
\newtheorem{proposition}{Proposition}

\newtheorem{example}{Example}
\newtheorem{theorem}{Theorem}
\newtheorem{lemma}{Lemma}

\newcommand{\LIA}{\mathtt{LIA}}
\newcommand{\LangSC}{\mathcal{L}_{sc}}
\newcommand{\LangSCLIA}{\mathcal{L}_{sc}^{\LIA}}
\newcommand{\LangSCExt}{\mathcal{L}_{sc}^{\mathtt{EXT}}}
\newcommand{\BAT}{\mathcal{D}}

\newcommand{\BATExt}{\BAT^{\mathtt{EXT}}}
\newcommand{\initKB}{\BAT_{S_0}}

\newcommand{\predSet}{\mathcal{P}}
\newcommand{\funcSet}{\mathcal{F}}
\newcommand{\predFuncSet}{\mathcal{Q}}
\newcommand{\codingSym}{\mu}
\newcommand{\actSet}{\mathcal{A}}
\newcommand{\intNum}{\mathbb{Z}}

\newcommand{\bool}{\mathbb{B}}

\newcommand{\act}{\alpha}
\newcommand{\true}{\top}
\newcommand{\false}{\bot}
\newcommand{\set}[1]{\{ #1 \}}

\newcommand{\nmodels}{\not \models}

\newcommand{\Poss}{Poss}
\newcommand{\Exec}{Exec}

\newcommand{\subHis}{\sqsubseteq}

\newcommand{\Trans}{Trans} 
\newcommand{\Final}{Final}
\newcommand{\Do}{Do}
\newcommand{\blank}{\hspace*{6mm}}
\renewcommand{\implies}{\supset}
\newcommand{\sitDom}{\Delta_S}
\newcommand{\highBAT}{\BAT^h}
\newcommand{\lowBAT}{\BAT^l}
\newcommand{\highModel}{M^h}
\newcommand{\lowModel}{M^l}
\newcommand{\forget}{\mathtt{forget}}
\newcommand{\theory}{\mathcal{T}}
\newcommand{\identical}{\underline{\leftrightarrow}}
\newcommand{\bisimilar}{\sim}
\newcommand{\countOper}{\#}

\title{A Syntactic Approach to Computing Complete and Sound Abstraction \\ in the Situation Calculus}
\author {
    Liangda Fang\textsuperscript{\rm 1, 4},
    Xiaoman Wang\textsuperscript{\rm 1},
    Zhang Chen\textsuperscript{\rm 1},
    Kailun Luo\textsuperscript{\rm 2},
    Zhenhe Cui\textsuperscript{\rm 3},
    Quanlong Guan\textsuperscript{\rm 1}\thanks{Corresponding author}
}
\affiliations {
    \textsuperscript{\rm 1}Jinan University, Guangzhou 510632, China \\
    \textsuperscript{\rm 2}Dongguan University of Technology, Dongguan 523808, China \\
    \textsuperscript{\rm 3}Hunan University of Science and Technology, Xiangtan 411201, China \\
    \textsuperscript{\rm 4}Pazhou Lab, Guangzhou 510330, China \\
    fangld@jnu.edu.cn, wangxm@stu2022.jnu.edu.cn, ztchen@stu2020.jnu.edu.cn, \\
    luokl@dgut.edu.cn, zhhcui@hnust.edu.cn, gql@jnu.edu.cn}

\usepackage{bibentry}

\begin{document}

\maketitle

\begin{abstract}
	\looseness=-1
	Abstraction is an important and useful concept in the field of artificial intelligence.
	To the best of our knowledge, there is no syntactic method to compute a sound and complete abstraction from a given low-level basic action theory and a refinement mapping.
	This paper aims to address this issue.
	To this end, we first present a variant of situation calculus, namely linear integer situation calculus, which serves as the formalization of high-level basic action theory.
	We then migrate \citeauthor{BanihashmemiGL2017}'s abstraction framework to one from linear integer situation calculus to extended situation calculus.
	Furthermore, we identify a class of Golog programs, namely guarded actions, so as to restrict low-level Golog programs, and impose some restrictions on refinement mappings. 
	Finally, we design a syntactic approach to computing a sound and complete abstraction from a low-level basic action theory and a restricted refinement mapping.
\end{abstract}
    \begin{links}
	\link{Extended version}{http://arxiv.org/abs/2412.11217}
    \end{links}

\section{Introduction}

\looseness=-1
Abstraction plays an important role in many fields of artificial intelligence, including planning, multi-agent systems and reasoning about actions.
The idea behind abstraction is to reduce a problem over large or even infinite state space to a problem over small and finite state space by aggregating similar concrete states into abstract states, and hence facilitate searching a solution for planning problems and checking some properties against multi-agent systems.

\looseness=-1
For classical planning, abstraction is often used as a heuristic function to guide the search of the solution \cite{HelmertHHN2014,SeippH2018}.
Generalized planning aims to find a uniform solution for possibly infinitely many problem instances \cite{Lev2005,SegJJ2018}. 
\citet{SrivastavaIZ2011} developed a method to generate a finite-state automata-based solution for generalized planning based on abstraction using $3$-valued semantics \cite{SagRW2002}.

\looseness=-1
In multi-agent systems, abstraction is used to accelerate verification of specifications formalized in computation-tree logic \cite{ShohamG2007}, alternating-time temporal logic \cite{LomuscioM2014,BelFF2023} and its epistemic extensions \cite{LomAM2016,BelFL2016}, strategy logic \cite{BelardinelliJM2023}, and $\mu$-calculus \cite{BalTK2006,GruOL2007}.

\looseness=-1
In the area of reasoning about actions, the situation calculus \cite{Reiter2001} is an expressive framework for modeling and reasoning about dynamic changes in the world.
Golog \cite{LevRLLS1997} and its concurrency extension ConGolog \cite{GiacomoLL2000} are agent programming languages based on the situation calculus that are popular means for the control of autonomous agents.
\citet{BanihashmemiGL2017} developed a general abstraction framework based on the situation calculus and ConGolog programming language.
This abstraction framework relates high-level basic action theory (BAT) to low-level BAT via the notion of refinement mapping.
A refinement mapping associates each high-level fluent to a low-level formula and each high-level action to a low-level ConGolog program.
Based on the concept of bisimulation between high-level models and low-level models \cite{Milner1971}, they also provide the definition of sound/complete abstractions of low-level BATs under refinement mappings.

\looseness=-1
Although a lot of efforts have made on abstractions, to the best of our knowledge, there is no syntactic approach to computing complete and sound abstractions given suitable refinement mappings.
This paper aims to address this issue.
To this end, we first present a variant of situation calculus, namely linear integer situation calculus, using linear integer arithmetic as the underlying logic.
We then adopt linear integer situation calculus as the formalization of high-level action theory and extend \citet{BanihashmemiGL2017}'s abstraction framework.
In addition, we identify a class of Golog programs, namely guarded actions, so as to restrict low-level Golog programs and impose some restrictions on refinement mappings.
Finally, given a low-level BAT and a restricted refinement mapping, we devise a syntactic approach to computing the initial knowledge base, the action precondition axioms and the successor state axioms of a complete and sound abstraction.

\section{Preliminaries}
\looseness=-1
In this section, we briefly introduce the concepts of linear integer arithmetic, the situation calculus, Golog programming language and forgetting.

\subsection{Linear Integer Arithmetic}
\looseness=-1
Let $\bool$ be the set of Boolean constants $\set{\true, \false}$ and $\intNum$ the set of integers.
The syntax of $\LIA$-definable terms and of formulas are defined as:
\begin{center}
	$e ::= c \mid x \mid e_1 + e_2 \mid  e_1 - e_2$ \\
	$\phi ::= \true \mid \false \mid p \mid e_1 = e_2 \mid e_1 < e_2 \mid e_1 \cong_c e_2 \mid \neg \phi \mid \phi_1 \land \phi_2 \mid \phi_1 \lor \phi_2$
\end{center}

\noindent
where $c \in \intNum$, $p$ is a Boolean variable and $x$ is a numeric variable.

\looseness=-1
The formula $e_1 \cong_{c} e_2$ means that $e_1$ and $e_2$ are congruence modulo $c$, that is, $e_1 - e_2$ is divisible by $c$.
The formula $\phi_1 \implies \phi_2$ is the shorthand for $\neg \phi_1 \lor \phi_2$, $e_1 \leq e_2$ for $e_1 = e_2 \lor e_1 < e_2$, and $e_1 \not \cong_c e_2$ for $\neg (e_1 \cong_c e_2)$.

\looseness=-1
We remark that the above syntax of $\LIA$ does not involve quantifier symbols.
This is because $\LIA$ admits quantifier elimination, that is, any $\LIA$-definable formula with quantifiers can be equivalently transformed into a quantifier-free one \cite{Coo1972,Mon2010}.

\subsection{Situation Calculus}

\looseness=-1
The standard situation calculus $\LangSC$ \cite{Reiter2001} is a first-order logic language with limited second-order features for representing and reasoning about dynamically changing worlds.
There are three disjoint sorts: \textit{action} for actions, \textit{situation}	for situations, and \textit{object} for everything else.
The constant $S_0$ denotes the initial situation before any action is performed; $do(a, s)$ denotes the successor situation resulting from performing action $a$ in situation $s$; and a binary predicate $s \subHis s'$ means that situation $s$ is a subhistory of situation $s'$.
The predicate $\Poss(a, s)$ describes the precondition of action $a$ in situation $s$.
Predicates and functions whose values may change from one situation to another are called \textit{fluent}, taking a situation term as their last argument.
For simplicity, we assume that there are no functional fluents in the standard situation calculus.

\looseness=-1
A formula $\phi$ with all situation arguments suppressed is called \textit{situation-suppressed}, and $\phi[s]$ denotes the uniform formula obtained from $\phi$ by restoring situation variable $s$ into all fluent names mentioned in $\phi$.
We say the formula $\phi[s]$ is an $\LIA$-definable formula uniform in $s$, if $\phi$ is $\LIA$-definable.
The notion of situation-suppressed and uniform terms can be similarly defined.

\looseness=-1
A basic action theory (BAT) $\BAT$ which describes how the world changes as the result of the available actions consists of the following disjoint sets of axioms: foundational axioms $\Sigma$, initial knowledge base (KB) $\BAT_{S_0}$\footnote{For simplicity, we assume that $\BAT_{S_0}$ is a formula uniform in $S_0$ rather than a set of formulas.}, action precondition axioms $\BAT_{ap}$, successor state axioms (SSAs) $\BAT_{ss}$, unique names axioms $\BAT_{una}$ for actions.

\looseness=-1
We introduce an abbreviation: $\Exec(s) \doteq \forall a, s^*. do(a, s^*) \sqsubseteq s \supset \Poss(a, s^*)$.
The notation $\Exec(s)$ means $s$ is an executable situation, that is, an action history in which it is possible to perform the actions one after the other.

\subsubsection{Extension of the Situation Calculus}
\looseness=-1
First-order logic is the underlying logic of the standard situation calculus.
However, it lacks the mechanism for unbounded iteration \cite{Gra1991} and has the limited counting ability \cite{KusS2017}.
It is necessary to represent unbounded iteration and counting ability in some domains.
To represent these domains in the situation calculus, \citet{CuiLL2021} use first-order logic with transitive closure and counting as the underlying logic.

\looseness=-1
Let $\phi(\vec{x}, \vec{y})$ be a formula with two $k$-tuples $\vec{x}$ and $\vec{y}$ of free variables, and $\vec{u}$ and $\vec{v}$ two $k$-tuples of terms.
$[\mathtt{TC}_{\vec{x}, \vec{y}}\phi(\vec{x}, \vec{y})](\vec{u}, \vec{v})$ is the formula which means the pair $(\vec{u}, \vec{v})$ is contained in the reflexive transitive closure of the binary relation on $k$-tuples that is defined by $\phi$.
The notation $\vec{x} = \vec{y}$ abbreviates for $\bigwedge_{i = 1}^k (x_i = y_i)$ where $x_i$ and $y_i$ are the $i$-th element of $\vec{x}$ and $\vec{y}$, respectively.
Similarly, $\vec{x} \neq \vec{y}$ abbreviates for $\bigvee_i^k (x_i \neq y_i)$.
For simplicity, we use $P^*(\vec{x}, \vec{y})$ for $[\mathtt{TC}_{\vec{x}, \vec{y}}P(\vec{x}, \vec{y})](\vec{x}, \vec{y})$ and $P^+(\vec{x}, \vec{y})$ for $P^*(\vec{x}, \vec{y}) \land \vec{x} \neq \vec{y}$ where $P(\vec{x}, \vec{y})$ is a $2k$-arity predicate.
Let $\phi(\vec{x})$ be a formula with a tuple $\vec{x}$ of free variables.
$\countOper \vec{x}. \phi(\vec{x})$ is a counting term denoting the number of tuples $\vec{x}$ satisfying the formula $\phi$.

\looseness=-1
To extend the situation calculus with counting, we introduce a sort \textit{integer} for integers. 
If $\phi(\vec{x})$ is a formula, then $\countOper\vec{x}. \phi(\vec{x})$ is a term of sort integer, with the same meaning in first-order logic with transitive closure and counting.
A term of sort integer is \textit{closed}, iff every variable is bounded by the existential quantifier $\exists$ or the counting operator $\countOper$.
In addition, we make the assumption that there are finitely many non-number objects.
To do this, we adopt \citet{LiL2020}'s approach.
They introduced an extra integer function symbol $\mu$ to represent a coding of objects into natural numbers.
The extended basic action theory $\BATExt$ of $\LangSCExt$ contains one additional axiom, namely \textit{finitely many objects axiom}, $\BAT_{fma}$, that is, the conjunction of the following sentences:
\begin{enumerate}
	\item $\forall x, y (\codingSym(x) = \codingSym(y) \implies x = y)$;	
	\item $\exists m, n \forall x (m \leq \codingSym(x) \land \codingSym(x) \leq n)$\footnote{The second conjunction in this paper is slightly different from \citet{LiL2020}'s. We introduce a sort \textit{integer} for integers in the situation calculus while they introduced a sort for natural numbers. Hence, our axiom needs to describe the existence of the smallest coding of objects.}.
\end{enumerate}
The first conjunction means that different objects have different codings while the second component says that there are the smallest coding and the largest one.

\looseness=-1
We explicitly provide the state constraints $\BAT_{con}$ that is a sentence of the form $\forall s. \Exec(s) \implies \phi[s]$.
We use the notation $\BAT^{-}_{con}$ for the situation-suppressed formula $\phi$ in $\BAT_{con}$.

\looseness=-1
We hereafter give an extended BAT $\BAT^{BW_l}$ for the blocks world \cite{SlaneyT2001,CookL2003}.
\begin{example} \label{exm:blocksworld}
	In blocks world, there is one table, one gripper and finitely many blocks.
	Each block is either on the table or on the other block.
	The gripper holds at most one block at one time.
	When the gripper is empty, it can pick up a block on which there is no block.
	The gripper can put down the block which it holds on the table or the other block.
	There is a special block $C$.
	Initially, the gripper is empty and at least one block is above $C$.
	The following is the formalization of blocks world.
	Throughout this paper, we assume that free variables are implicitly universally quantified.

	\textbf{Fluents:} 
	\begin{itemize}
		\item $holding(x, s)$: the gripper holds block $x$ in situation $s$;
		\item $on(x, y, s)$: block $x$ is on block $y$ in situation $s$;
	\end{itemize}
	
	\textbf{Actions:} 
	\begin{itemize}
		\item $pickup(x)$: pick up block $x$ from the table;
		
		\item $putdown(x)$: put down block $x$ onto the table;
		
		\item $unstack(x, y)$: pick up block $x$ from block $y$;
		
		\item $stack(x, y)$: put block $x$ onto block $y$;
	\end{itemize}
	
	\textbf{Initial knowledge base $\BAT^{BW_l}_{S_0}$:} \\
	\hspace*{4mm} $\neg \exists x. holding(x, S_0) \land \exists x. on^{+}(x, C, S_0) \land \BAT^{BW_l-}_{con}[S_0]$;
	
	\textbf{Action precondition axioms $\BAT^{BW_l}_{ap}$:}           
	\begin{itemize}
		\item $Poss(pickup(x), s) \equiv \\ \blank \neg (\exists y) on(y, x, s) \land \neg (\exists y) holding(y, s)$;
		
		\item $Poss(putdown(x), s) \equiv holding(x, s)$;
		
		\item $Poss(unstack(x, y), s) \equiv \\ \blank on(x, y, s) \land \neg (\exists z) on(z, x, s) \land \neg (\exists z) holding(z, s)$;
		
		\item $Poss(stack(x, y), s) \equiv \\ \blank holding(x, s) \land \neg \exists z. on(z, y, s)$;
	\end{itemize}
	
	\textbf{Successor state axioms $\BAT^{BW_l}_{ss}$:}
	\begin{itemize}
		\item $on(x, y, do(a, s)) \equiv [a = stack(x, y)] \lor \\ \blank [on(x, y, s) \land a \neq unstack(x, y)]$;
		
		\item $holding(x, do(a, s)) \equiv \\ \blank [a = pickup(x) \lor \exists y (a = unstack(x, y))] \lor  \\ \blank [holding(x, s) \land a \neq putdown(x) \land \\ \blank \hspace*{1mm} \forall y  (a \neq stack(x, y))]$;

	\end{itemize}
	
	\textbf{State constraints $\BAT^{BW_l-}_{con}$:} the conjunction of the following sentences

	\begin{itemize}
		\item $\neg on^+(x, x)$;
		\item $on^+(x, y) \land on^+(x, z) \implies [y = z \lor on^+(y, z) \lor on^+(z, y)]$;
		\item $on^+(y, x) \land on^+(z, x) \implies [y = z \lor on^+(y, z) \lor on^+(z, y)]$;
		\item $holding(x) \land holding(y) \implies x = y$;
		\item $holding(x) \implies \forall y [\neg on(x, y) \land \neg on(y, x)]$;
		\item $on(x, y) \implies \neg holding(x) \land \neg holding(y)$. \qed
	\end{itemize} 
\end{example}

\looseness=-1
The first three sentences of $\BAT^{BW_l-}_{con}$ together say that the predicate $on(x, y)$ defines a collection of linear orders.
The four sentence requires that $holding(x)$ is \textit{exclusive}, that is, at most one block is holding.
The last two sentences mean that the two predicates $holding(x)$ and $on(x, y)$ are \textit{mutex}, that is, (1) if a block $x$ is holding, then no block $y$ is on or under block $x$, and (2) if a block $x$ is on block $y$, then neither of them is holding.

\looseness=-1
Our state constraints are slightly different from those in \cite{CookL2003}.
First, they used a primitive predicate $above(x, y)$ to represent block $x$ is above block $y$ with the same meaning of the transitive closure $on^+(x, y)$ we use in this paper.
The two state constraints $above(x, y) \land above(y, z) \implies above(x, z)$ and $above(x, y) \implies [(\exists z) on(x, z) \lor (\exists w) on(w, y)]$, presented in \cite{CookL2003}, are redundant.
Second, the finitely many objects axiom implies that every tower has a bottom and a top block.
The two abbreviations $ontable \doteq \neg (\exists y) above(x, y)$ and $clear \doteq \neg (\exists y) above(y, x)$ and the two state constraints $ontable(x) \lor \exists y (above(x, y) \land ontable(y))$ and $clear(x) \lor \exists y (above(y, x) \land clear(y))$, defined in \cite{CookL2003}, are unnecessary.
Third, blocks world in this paper involves an additional predicate $holding(x)$.
We need three extra state constraints for $holding(x)$.

\subsection{Golog}
\looseness=-1
To represent and reason about complex actions, \citet{LevRLLS1997} introduced a programming language, namely Golog.
The syntax of Golog programs is defined as:
\begin{center}
	$\delta ::= nil \mid a \mid \phi? \mid \delta_{1};\delta_{2} \mid  \delta_{1} | \delta_{2} \mid \pi x.\delta \mid \delta^{*}$
\end{center}
where $nil$ is an empty program; $a$ is an action term, possibly with variables; $\phi?$ is a test action which tests whether $\phi$ holds in the current situation where $\phi$ is a situation-suppressed formula; $\delta_{1}; \delta_{2}$ is a sequential structure of $\delta_{1}$ and $\delta_{2}$; $\delta_{1} \mid \delta_{2}$ is the non-deterministic choice between $\delta_{1}$ and $\delta_{2}$; $\pi x. \delta$ is the program $\delta$ with some non-deterministic choice of a legal binding for variable $x$; $\delta^{*}$ means $\delta$ executes zero or more times.
A Golog program is \textit{closed}, iff every variable is bounded by the existential quantifier $\exists$, the counting operator $\countOper$, or the choice operator $\pi$.

\looseness=-1
The semantics of Golog programs is specified by the notation $\Do(\delta, s, s')$, which means that $s'$ is a legal terminating situation of an execution of $\delta$ starting in situation $s$.
There are two ways to define $\Do(\delta, s, s')$.
\citet{LevRLLS1997} treated Golog programs $\delta$ as an additional extralogical symbols and hence $\Do(\delta, s, s')$ is expanded into a formula by induction on the structure of $\delta$.
However, the above treatment only works for programs without concurrency and cannot qualify over programs.
\citet{GiacomoLL2000} encoded ConGolog programs, the concurrency extension to Golog, as terms in the situation calculus, and define $\Do(\delta, s, s')$ as an abbreviation: 
\begin{center}
	$\Do(\delta, s, s') \doteq \exists \delta'.\Trans^{*}(\delta, s, \delta', s') \land \Final(\delta', s')$
\end{center}
where $\Trans(\delta, s, \delta', s')$ means that a single step of program $\delta$ starting in situation $s$ leads to situation $s'$ with the remaining program $\delta'$ to be executed, $\Trans^{*}$ denotes the reflexive transitive closure of $\Trans$, and $\Final(\delta, s)$ means that program $\delta$ may legally terminate in situation $s$.

\looseness=-1
The semantics of both approaches are equivalent for Golog programs (Theorem 1 in \cite{GiacomoLL2000}).
In this paper, we adopt \citet{LevRLLS1997}'s approach since we use only Golog programs and do not need to qualify over programs.

\subsection{Forgetting}
\looseness=-1
The concept of forgetting dates back to \cite{Boole1854}, who first considered forgetting in propositional logic.
\citet{LinR1994} studied forgetting in first-order logic.

\begin{definition} \label{def:identical} \rm 
	Let $\predFuncSet$ be a set of predicate and functional symbols.
	We say two structures $M_1$ and $M_2$ are $\predFuncSet$-identical, written $M_1 \identical_{\predFuncSet} M_2$, if $M_1$ and $M_2$ agree on everything except possibly on the interpretation of every symbol of $\predFuncSet$.
\end{definition}

\begin{definition}  \label{def:forget} \rm 
	Let $\theory$ be a theory and $\predFuncSet$ a set of predicate and functional symbols. 
	A theory $\theory'$ is a result of forgetting $\predFuncSet$ in $\theory$, written $\forget(\theory, \predFuncSet) \equiv \theory'$, if for every structure $M'$, $M' \models \theory'$ iff there is a model $M$ of $\theory$ s.t. $M \identical_{\predFuncSet} M'$.
\end{definition}

\looseness=-1
Forgetting a set of symbols results in a weak theory that has the same set of logical consequences irrelevant to the forgotten symbols as the original theory.

\begin{proposition}  \label{prop:forget}
	Let $\theory$ be a theory and $\predFuncSet$ a set of predicate and functional symbols.
	Then, $\theory \models \forget(\theory, \predFuncSet)$, and for any sentence $\phi$ wherein $\predFuncSet$ does not appear, $\theory \models \phi$ iff $\forget(\theory, \predFuncSet) \models \phi$.
\end{proposition}

\looseness=-1
If two formulas $\phi$ and $\psi$ are equivalent under a background theory $\theory$, then the result of forgetting a set of symbols in the union of  $\theory$ and $\psi$ and that in the union $\theory$ and $\psi$ are equivalent.

\begin{proposition} \label{prop:forgetUnderTheory}
	Let $\theory$ be a theory, $\phi$ and $\psi$ two formulas, and $\predFuncSet$ a set of predicate and functional symbols s.t. $\theory \models \phi \equiv \psi$.
	Then, $\forget(\theory \cup \set{\phi}, \predFuncSet) \equiv \forget(\theory \cup \set{\psi}, \predFuncSet)$.
\end{proposition}

\section{Linear Integer Situation Calculus}
\looseness=-1
In this section, we propose a variant of the situation calculus based on $\LIA$, namely linear integer situation calculus, that serves as a formalization of high-level BATs.
The main idea is to impose a syntactic constraint of $\LIA$ on the situation calculus.
In the following, we use $\LangSCLIA$ to denote the language of linear integer situation calculus.

\looseness=-1
Linear integer situation calculus $\LangSCLIA$ involves only three sorts integers, actions and situations and allows functional fluents.
In addition, $\LangSCLIA$ obeys the following restrictions:
(1) every predicate and functional fluent contains exactly one argument of sort situation;
(2)	the domain of every functional fluent is the set of integers; 
(3)	every action function has no argument, and hence is a ground action.

\looseness=-1
The linear integer basic action theory (LIBAT) $\BAT^{\LIA} = \Sigma \cup \BAT_{S_0} \cup \BAT_{ap} \cup \BAT_{ss} \cup \BAT_{una}$ is 
\begin{itemize}
\item The foundational axioms for situations $\Sigma$ and the set of unique names axioms for actions $\BAT_{una}$ are the same as those in the standard BAT $\BAT$.

\item Initial knowledge base $\BAT_{S_0}$: a formula $\phi[S_0]$ that is an $\LIA$-definable formula uniform in $S_0$.

\item Action precondition axioms $\BAT_{ap}$: for each action $A$, there is an axiom of the form $\Poss(A, s) \equiv \Pi_A[s]$, where $\Pi_A[s]$ is an $\LIA$-definable formula uniform in $s$.

\item Successor state axioms $\BAT_{ss}$: for each predicate fluent $P$, there is one axiom of the form $P(do(a, s)) \equiv \Phi_P(a, s)$ where $\Phi_P(a, s)$ is defined as:
\begin{center}
    $[\bigvee \limits_{i = 1}^{m} (a = \act^+_i \land \gamma^+_{i}[s]) ] \lor [P(s) \land \neg (\bigvee \limits_{j = 1}^{n} (a = \act^-_j \land \gamma^-_{j}[s]))]$
\end{center}
where each $\gamma^+_{i}[s]$ and $\gamma^-_{j}[s]$ is an $\LIA$-definable formula uniform in $s$;

and for each functional fluent $f$, there is one axiom of the form $f(do(a, s)) = y \equiv \Phi_f(y, a, s)$ where $\Phi_f(y, a, s)$ is defined as:

\begin{center}
$[\bigvee\limits_{i = 1}^{m} (y = t_{i}[s] \land a = \act^+_i \land \gamma^+_{i}[s])] \lor$ \\ 
$[y = f(s) \land \neg (\bigvee\limits_{j = 1}^{n} (a = \act^-_j \land \gamma^-_{j}[s]))]$
\end{center}

where each $\gamma^+_{i}[s]$ and $\gamma^-_{j}[s]$ is an $\LIA$-definable formula uniform in $s$ and each $t_{i}[s]$ is an $\LIA$-definable term uniform in $s$.

\end{itemize}

\looseness=-1
Given an action $\alpha$, $\Phi_P(\alpha, s)$ can be simplified as a formula uniform in $s$, denoted by $\tilde{\Phi}_{P, \alpha}[s]$, via utilizing the unique names axioms for actions to remove the comparability relation between actions.
Similarly, we use $\tilde{\Phi}_{f, \alpha}(y)[s]$ to denote the formula uniform in $s$ obtained from $\Phi_f(y, \alpha, s)$ via simplification.

We close this section by providing a LIBAT $\BAT^{BW_h}$ of blocks world.
\begin{example} \label{exm:BATLIA}
\looseness=-1

\textbf{Fluents:}
\begin{itemize}
	\item $Holding(s)$: the gripper holds a block in situation $s$;
	
	\item $Num(s)$: the number of blocks that are above block $C$ in situation $s$;
\end{itemize}

\textbf{Actions:} 
\begin{itemize}
	\item $PickAboveC$: pick up the topmost block on the tower that contains block $C$;
	
	\item $Putdown$: put down the block which the gripper holds onto the table;
\end{itemize}

\textbf{Initial knowledge base $\BAT^{BW_h}_{S_0}$:} \\ 
\hspace*{4mm} $\neg Holding(s) \land Num(s) > 0$;

\textbf{Action precondition axioms $\BAT^{BW_h}_{ap}$:}        
\begin{itemize}
\item $Poss(PickAboveC, s) \equiv \\ \blank \neg Holding(s) \land Num(s) > 0$;		

\item $Poss(Putdown, s) \equiv Holding(s)$;
\end{itemize}

\textbf{Successor state axioms $\BAT^{BW_h}_{ss}$:}
\begin{itemize}
\item $Holding(do(a, s)) \equiv [a = PickAboveC] \lor \\ \blank [Holding(s) \land a \neq Putdown]$;

\item $Num(do(a, s)) = y \equiv \\ \blank [y = Num(s) - 1 \land a = PickAboveC] \lor \\ \blank [y = Num(s) \land a \neq PickAboveC]$. 
\qed
\end{itemize}
\end{example}

\section{An Abstraction Framework from $\LangSCLIA$ to $\LangSCExt$}
\looseness=-1
In this section, based on \citet{BanihashmemiGL2017}'s work, we develop an abstraction framework from linear integer situation calculus to extended situation calculus.

\looseness=-1
We assume that we are given two BATs $\BAT^l$ and $\BAT^h$.
We consider $\BAT^l$ formulated in $\LangSCExt$ as the low-level BAT and $\BAT^h$ formulated in $\LangSCLIA$ as the high-level BAT.
The low-level BAT $\BAT^l$ contains a finite set of action types $\actSet^l$ and a finite set of predicate fluents $\predSet^l$.
Given a low-level situation-suppressed formula $\phi$, we use $\predSet^l(\phi)$ for the set of predicate fluents that occur in $\phi$.
The meaning of the two notations $\actSet^h$ and $\predSet^h$ are similar for the high-level BAT $\BAT^h$, and $\funcSet^h$ stands for a finite set of high-level functional fluents.
Given a high-level situation-suppressed formula $\phi$, we use $\predSet^h(\phi)$ (resp. $\funcSet^h(\phi)$) for the set of predicate (resp. functional) fluents that occur in $\phi$.

\looseness=-1
The concept of refinement mapping from $\highBAT$ to $\lowBAT$ is the same as that proposed by \citet{BanihashmemiGL2017} except the following. 
(1) We allow functional fluents in high-level BATs.
(2) The refinement mapping assigns every high-level action to a closed low-level Golog program since every high-level action has no argument.
(3) The refinement mapping assigns every high-level predicate fluent to a closed low-level formula since every high-level predicate fluent contains exactly one argument of sort situation.
(4) The refinement mapping assigns every high-level functional fluent to a closed counting term since the image of every high-level functional fluent is the set of integers.

\begin{definition}\label{def:refMapp} \rm
We say that a function $m$ is a refinement mapping from $\BAT^h$ to $\BAT^l$, iff
\begin{enumerate}
\item for every high-level action $\act$, $m(\act) \doteq \delta_\act$ where $\delta_\act$ is a closed Golog program defined over $\actSet^l$ and $\predSet^l$;

\item for every high-level predicate fluent $P$, $m(P) \doteq \phi_{P}$ where $\phi_{P}$ is a closed situation-suppressed formula defined over $\predSet^l$;

\item for every high-level functional fluent $f$, $m(f) \doteq \countOper \vec{x}. \phi_f(\vec{x})$ where $\phi_f$ is a situation-suppressed formula with a tuple $\vec{x}$ of free variables defined over $\predSet^l$.
\end{enumerate}
\end{definition}

\looseness=-1
Given a refinement mapping $m$ and a high-level situation-suppressed formula $\phi$, $m(\phi)$ denotes the result of replacing every occurrence of high-level predicate fluent $P$ and functional fluent $f$ by the low-level counterpart $m(P)$ and $m(f)$, respectively.
The mapping formula $\Psi_m$ based on $m$, is defined as $[\bigwedge_{P \in \predSet^h} P \equiv m(P) ] \land [\bigwedge_{f \in \funcSet^h} f = m(f)]$.
Intuitively, $\Psi_m$ says that every high-level predicate fluent has the same Boolean value as its corresponding low-level closed formula defined by $m$, and similarly for every high-level functional fluent.

\looseness=-1
\begin{example} \label{exm:refMap}
A refinement mapping $m$ from $\BAT^{BW_h}$ to $\BAT^{BW_l}$ is defined as:
\begin{itemize}
\item $m(Num) \doteq \countOper x. on^{+}(x, C)$;

\item $m(Holding) \doteq \exists x. holding(x)$;

\item $m(PickAboveC) \doteq (\pi x, y) on^+(x, C)?; unstack(x, y)$;

\item $m(Putdown) \doteq \pi x. putdown(x)$.  \qed
\end{itemize}
\end{example}

\looseness=-1
We hereafter introduce the concept of $m$-isomorphism between high-level situations and low-level situations.
We assume that $M^h$ and $M^l$ are models of $\highBAT$ and $\lowBAT$, respectively.
The two assignments $v[s/s_{h}]$ and $v[s/s_l]$ are the same as $v$ except mapping variable $s$ to situation $s_{h}$ and $s_l$, respectively.

\begin{definition}\label{def:isomorphic} \rm
Given a refinement mapping $m$, we say that situation $s_h$ in $M^h$ is $m$-isomorphic to situation $s_l$ in $M^l$, written $s_h \bisimilar_{m}^{M^h, M^l} s_l$, iff
\begin{itemize}
	\item for every high-level predicate fluent $P$ and every variable assignment $v$, $M^h, v[s/s_h] \models P(s)$ iff $M^l, v[s/s_l] \models m(P)[s]$.

	\item for every high-level functional fluent $f$ and every variable assignment $v$, $M^h, v[s/s_h] \models f(s) = y$ iff $M^l, v[s/s_l] \models m(f)[s] = y$.
\end{itemize}
\end{definition}

\looseness=-1
We can get the following results:
\begin{proposition} \label{prop:isomorphic}
	If $s_h \bisimilar_{m}^{M^h, M^l} s_l$, then for every high-level situation-suppressed formula $\phi$ and every variable assignment $v$, $M^h, v[s/s_h] \models \phi[s]$ iff $ M^l, v[s/s_l] \models m(\phi)[s]$.
\end{proposition}

\begin{proposition} \label{prop:isomorphicForget}
	If $s_h \bisimilar_{m}^{M^h, M^l} s_l$, then for every low-level situation-suppressed formula $\phi$ and every variable assignment $v$, we have $M^l, v[s/s_l] \models \phi[s] \land \lowBAT_{fma}$ implies that $M^h, v[s/s_h] \models \forget(\phi \land \lowBAT_{fma} \land \Psi_m, \predSet^l \cup \set{\codingSym})[s]$.
\end{proposition}

\looseness=-1
To relate high-level and low-level theories, we resort to $m$-bisimulation relation between high-level models and low-level models.
The notation $\sitDom^{M}$ stands for the situation domain of $M$.

\begin{definition} \label{def:m-bisim} \rm
A relation $B \subseteq \sitDom^{M^h} \times \sitDom^{M^l}$ is an $m$-bisimulation between $M^h$ and $M^l$, iff $(s_h, s_l) \in B$ implies that

\begin{description}
	\item[atom] $s_h$ and $s_l$ are $m$-isomorphic.

	\item[forth] for every high-level action $\act$ and every low-level situation $s_l'$ s.t. $M^l, v[s/s_l, s'/s_l'] \models \Do(m(\act), s, s')$, there is a high-level situation $s_h'$ s.t. $M^h, v[s/s_h, s'/s_h'] \models \Poss(\act, s) \land s' = do(\act, s)$ and $(s'_h, s_l') \in B$.

	\item[back] for every high-level action $\act$ and every high-level situation $s_h'$ s.t. $M^h, v[s/s_h, s'/s_h'] \models \Poss(\act, s) \land s' = do(\act, s)$, there is a  low-level situation $s_l'$ s.t. $M^l, v[s/s_l, s'/s_l'] \models \Do(m(\act), s, s')$ and $(s'_h, s_l') \in B$.
\end{description}		
\end{definition}

\looseness=-1
We say $M^h$ is $m$-bisimilar to $M^l$, written $M^h \bisimilar_m M^l$, iff there is an $m$-bisimulation $B$ between $M^h$ and $M^l$ s.t. $(S_0^{M^h}, S_0^{M^l}) \in B$.

\looseness=-1
Based on the notion of $m$-bisimulation, we hereafter provide definitions of sound and complete abstraction of low-level theories.

\begin{definition}\rm

We say $\BAT^h$ is a \textit{sound abstraction} of $\BAT^l$ relative to a refinement mapping $m$, iff for every model $M^l$ of $\BAT^l$, there is a model $M^h$ of $\BAT^h$ s.t. $M^h \bisimilar_m M^l$.
\end{definition}

\looseness=-1
Let $m$ be a refinement mapping.
Let $\pi_m$ be a low-level program $[m(\act_1) | \cdots | m(\act_n)]$ where each $\act_i \in \actSet^h$.
Intuitively, $\pi_m$ denotes any refinement of any high-level primitive actions.
The iteration $\pi^*_m$ denotes any refinement of any high-level action sequence.

\begin{theorem} \label{thm:soundAbs}
	$\BAT^h$ is a sound abstraction of $\lowBAT$ relative to a refinement mapping $m$ iff the following hold
	\begin{enumerate}
		\item $\lowBAT_{S_0} \land \lowBAT_{fma} \models m(\phi)[S_0]$ where $\highBAT_{S_0} = \phi[S_0]$;
		
		\item $\lowBAT \models \forall s. \Do(\pi^*_m, S_0, s) \implies \\
		 \blank \bigwedge_{\act \in \ \actSet^h} (m(\Pi_{\act})[s] \equiv \exists s' \Do(m(\act), s, s'))$; 
		
		\item $\lowBAT \models \forall s. \Do(\pi^*_m, S_0, s) \implies \\
		\blank \bigwedge_{\act \in \ \actSet^h} \forall s' \{ \Do(m(\act), s, s') \implies \\
		\blank \blank \bigwedge_{P \in \predSet^h} (m(P)[s'] \equiv m(\tilde{\Phi}_{P, \act})[s]) \land \\
		\blank \blank \bigwedge_{f \in \funcSet^h} \forall y(m(f)[s'] = y \equiv m(\tilde{\Phi}_{f, \alpha})(y)[s]) \}$.
	\end{enumerate}
\end{theorem}

\begin{definition} \rm
	We say $\BAT^h$ is a \textit{complete abstraction} of $\BAT^l$ relative to a refinement mapping $m$, iff for every model $M^h$ of $\BAT^h$, there is a model $M^l$ of $\BAT^l$ s.t. $M^h \bisimilar_m M^l$.
\end{definition}

\looseness=-1
If the high-level BAT $\BAT^h$ is a sound abstraction of the low-level BAT $\lowBAT$, then deciding if $\highBAT$ is a complete abstraction of $\lowBAT$ reduces to determining if the high-level initial knowledge base $\highBAT_{S_0}$ is the result of forgetting all low-level predicate fluents and the coding symbol in the union of the low-level initial knowledge base, the finitely many objects axiom and the mapping formula.

\begin{theorem} \label{thm:compSoundAbs}
	Suppose that $\BAT^h$ is a sound abstraction of $\lowBAT$ relative to a refinement mapping $m$.
	Then, $\highBAT$ is a complete abstraction of $\lowBAT$ relative to $m$ iff $\highBAT_{S_0} \equiv \forget(\phi \land \lowBAT_{fma} \land \Psi_m, \predSet^l \cup \set{\codingSym})[S_0]$ where $\phi[S_0] = \lowBAT_{S_0}$.
\end{theorem}

\section{Computing Sound and Complete Abstraction}
\looseness=-1
In this section, we first identify a class of Golog programs, namely \textit{guarded actions}, and connect high-level actions to low-level guarded actions.
We then impose some restrictions on refinement mappings and provide a syntactic approach to computing the initial KB, action precondition axioms and successor state axioms of sound and complete abstraction from a low-level BAT and a restricted refinement mapping.

\subsection{Guarded Actions}
\looseness=-1
We say a Golog program $\delta$ is a \textit{guarded action}, if it is of the form $\pi \vec{x}. \psi(\vec{x})?; A(\vec{x}, \vec{c})$.
A guarded action requires all variables $\vec{x}$ of action $A(\vec{x}, \vec{c})$ to be guarded by the formula $\psi(\vec{x})$.

\begin{definition} \label{def:enabling} \rm
Let $A(\vec{x}, \vec{c})$ be a low-level action with a tuple $\vec{x}$ of variables and a tuple $\vec{c}$ of constants.
Let $\phi(\vec{y})$ and $\psi(\vec{x})$ be two low-level situation-suppressed formulas where $\vec{x}$ includes all of the variables of $\vec{y}$.
We say $A(\vec{x}, \vec{c})$ is 
\begin{itemize}		
	\item \textit{alternately} $\phi(\vec{y})$-\textit{enabling} under $\psi(\vec{x})$, iff $\BAT^l \models \forall s, \vec{x}. (\Poss(A(\vec{x}, \vec{c}), s) \land \psi(\vec{x})[s]) \implies (\neg \phi(\vec{y})[s] \land \phi(\vec{y})[do(A(\vec{x}, \vec{c}), s)])$;

	\item \textit{singly} $\phi(\vec{y})$-\textit{enabling} under $\psi(\vec{x})$, iff $A(\vec{x}, \vec{c})$ is alternately $\phi(\vec{y})$-\textit{enabling} under $\psi(\vec{x})$ and $\lowBAT \models \forall s, \vec{x}. (\Poss(A(\vec{x}, \vec{c}), s) \land \psi(\vec{x})[s]) \implies \forall \vec{z} (\vec{y} \neq \vec{z} \implies \phi(\vec{z})[s] \equiv \phi(\vec{z})[do(A(\vec{x}, \vec{c}), s)])$.
\end{itemize}
\end{definition}

\begin{definition} \label{def:invariant} \rm
	Let $A(\vec{x}, \vec{c})$ be a low-level action with a tuple $\vec{x}$ of variables and a tuple $\vec{c}$ of constants.
	Let $\phi(\vec{y})$ and $\psi(\vec{x})$ be two low-level situation-suppressed formulas. 
	We say $A(\vec{x}, \vec{c})$ is $\phi(\vec{y})$-\textit{invariant} under $\psi(\vec{x})$, iff $\BAT^l \models \forall s, \vec{x}. (\Poss(A(\vec{x}, \vec{c}), s) \land \psi(\vec{x})[s]) \implies \forall \vec{y}(\phi(\vec{y})[s] \equiv \phi(\vec{y})[do(A(\vec{x}, \vec{c}), s)])$.
\end{definition}

\looseness=-1
Alternative enabling property says that (1) $\phi(\vec{y})$ does not hold in the situation $s$ in which both the precondition of action $A(\vec{x}, \vec{c})$ and the condition $\psi(\vec{x})$ holds, and (2) $\phi(\vec{y})$ holds after performing $A(\vec{x}, \vec{c})$.
Single enabling property is a stronger property than alternative enabling, stipulating that $A(\vec{x}, \vec{c})$ does not affect the Boolean value of $\phi(\vec{z})$ for any tuple $\vec{z}$ of objects not equal to $\vec{y}$.
Invariance property means that $A(\vec{x}, \vec{c})$ does not affect the truth value of $\phi(\vec{y})$ with any tuple $\vec{y}$. 
\looseness=-1

\begin{definition} \label{def:exclusive} \rm
	Let $\phi(\vec{x})$ be a low-level situation suppressed formula.
	We say $\phi(\vec{x})$ is \textit{exclusive}, iff $\BAT^l \models \forall s, \vec{x}, \vec{y}. \Exec(s) \implies (\phi(\vec{x})[s] \land \phi(\vec{y})[s] \implies \vec{x} = \vec{y})$.
\end{definition}

\looseness=-1
Intuitively, at most one tuple $\vec{x}$ satisfies the formula $\phi(\vec{x})$ in every executable situation.

\begin{example} \label{exm:enablingInvariant}
	We continue with the low-level BAT $\BAT^{BW_l}$ of blocks world illustrated in Example \ref{exm:blocksworld}.
	For any situation $s$, the action precondition axiom for $putdown(x)$ entails $holding(x, s)$ and the successor state axioms entails $\neg holding(x, do(putdown(x), s))$.
	Therefore, $putdown(x)$ is alternately $\neg holding(x)$-enabling under $\true$.	
	In addition, performing action $putdown(x)$ does not change the Boolean value of $on^+(y, C)$ of any object $y$  and hence $putdown(x)$ is $on^+(y, C)$-invariant under $\true$.

	\looseness=-1
	Similarly, $unstack(x, y)$ is alternately $on^+(x, C)$-enabling under $on^+(x, C)$ and alternately $holding(x)$-enabling under $\true$.
	Furthermore, when $on^+(x, C)$ holds, performing $unstack(x, y)$ only grasps the topmost block on the tower contains block $C$ and does not catch other blocks on the same tower.
	Hence, $unstack(x, y)$ is singly $on^+(x, C)$-enabling under $on^+(x, C)$.

	\looseness=-1
	The fourth state constraint $holding(x) \land holding(y) \implies x = y$ deduces that $holding(x)$ is exclusive. \qed
\end{example}

\looseness=-1
We hereafter define connections between low-level guarded actions together with the low-level situation-suppressed formulas and its corresponding high-level actions together with high-level predicate (or functional) fluents.

\begin{definition} \label{def:highActEnabling} \rm
Let $m$ be a refinement mapping, $\act$ a high-level action, $P$ a high-level predicate fluent and $f$ a high-level functional fluent.
We say 
\begin{enumerate}
\item $\act$ is $P$-enabling relative to $m$, iff $m(\act) \doteq \pi \vec{x}. \psi(\vec{x})?; A(\vec{x}, \vec{c})$, $m(P) \doteq \exists \vec{y}. \phi(\vec{y})$ and $A(\vec{x}, \vec{c})$ is alternately $\phi(\vec{y})$-enabling under $\psi(\vec{x})$; 

\item $\act$ is $P$-disabling relative to $m$, iff $m(\act) \doteq \pi \vec{x}. \psi(\vec{x})?; A(\vec{x}, \vec{c})$, $m(P) \doteq \exists \vec{y}. \phi(\vec{y})$, $A(\vec{x}, \vec{c})$ is alternately $\neg \phi(\vec{y})$-enabling under $\psi(\vec{x})$ and $\phi(\vec{y})$ is exclusive; 

\item $\act$ is $P$-invariant relative to $m$, iff $m(\act) \doteq \pi \vec{x}. \psi(\vec{x})?; A(\vec{x}, \vec{c})$, $m(P) \doteq \exists \vec{y}. \phi(\vec{y})$ and $A(\vec{x}, \vec{c})$ is $\phi(\vec{y})$-invariant under $\psi(\vec{x})$; 

\item $\act$ is $f$-incremental relative to $m$, iff $m(\act) \doteq \pi \vec{x}. \psi(\vec{x})?; A(\vec{x}, \vec{c})$, $m(f) \doteq \countOper \vec{y}. \phi(\vec{y})$ and $A(\vec{x}, \vec{c})$ is singly $\phi(\vec{y})$-enabling under $\psi(\vec{x})$; 

\item $\act$ is $f$-decremental relative to $m$, iff $m(\act) \doteq \pi \vec{x}. \psi(\vec{x})?; A(\vec{x}, \vec{c})$, $m(f) \doteq \countOper \vec{y}. \phi(\vec{y})$  and $A(\vec{x}, \vec{c})$ is singly $\neg \phi(\vec{y})$-enabling under $\psi(\vec{x})$; 

\item $\act$ is $f$-invariant relative to $m$, iff $m(\act) \doteq \pi \vec{x}. \psi(\vec{x})?; A(\vec{x}, \vec{c})$, $m(f) \doteq \countOper \vec{y}. \phi(\vec{y})$  and $A(\vec{x}, \vec{c})$ is $\phi(\vec{y})$-invariant under $\psi(\vec{x})$. 
\end{enumerate}
\end{definition}

\looseness=-1
The following proposition provides the effect of executing high-level action $\act$ with properties illustrated in Definition \ref{def:highActEnabling} when the high-level BAT is a complete abstraction of the low-level BAT.

\begin{proposition} \label{prop:guardedAct}
	Suppose that $\BAT^h$ is a complete abstraction of $\BAT^l$ relative to $m$.
	Let $\act$ be a high-level action.
	Then, the following holds:
	\begin{enumerate}
		\item If $\act$ is $P$-enabling relative to $m$, then $\BAT^h \models \Exec(s) \land \Poss(\act, s) \implies P(do(\act, s))$;

		\item If $\act$ is $P$-disabling relative to $m$, then $\BAT^h \models \Exec(s) \land \Poss(\act, s)  \implies \neg P(do(\act, s))$;

		\item If $\act$ is $P$-invariant relative to $m$, then $\BAT^h \models \Exec(s) \land \Poss(\act, s) \implies P(do(\act, s)) \equiv P(s)$;

		\item If $\act$ is $f$-incremental relative to $m$, then $\BAT^h \models \Exec(s) \land \Poss(\act, s) \implies f(do(\act, s)) = f(s) + 1$;

		\item If $\act$ is $f$-decremental relative to $m$, then $\BAT^h \models \Exec(s) \land \Poss(\act, s) \implies f(do(\act, s)) = f(s) - 1$;

		\item If $\act$ is $f$-invariant relative to $m$, then $\BAT^h \models \Exec(s) \land \Poss(\act, s) \implies f(do(\act, s)) = f(s)$.
	\end{enumerate}
\end{proposition}

\begin{example} \label{exm:highActEnabling}
	Example \ref{exm:BATLIA} presents high-level BAT $\BAT^{BW_h}$ of blocks world, which will be later verified as a sound and complete abstraction of $\BAT^{BW_l}$.
	Example \ref{exm:refMap} illustrates a refinement mapping $m$ where $m(Holding) \doteq \exists x. holding(x)$ and $m(Putdown) \doteq \pi x. putdown(x)$.	
	Example \ref{exm:enablingInvariant} shows that $putdown(x)$ is alternately $\neg holding(x)$-enabling under $\true$ and $holding(x)$ is exclusive.
	Hence, high-level action $Putdown$ is $Holding$-disabling relative to $m$.
	By Proposition \ref{prop:guardedAct}, $Holding$ does not hold by performing action $Putdown$ in every high-level executable situation.
		
	Similarly, $Putdown$ is $Num$-invariant relative to $m$ while $PickAboveC$ is $Holding$-enabling and $Num$-decremental.	
	By Proposition \ref{prop:guardedAct}, $Holding$ becomes true and the value of $Num$ decrements by $1$ via performing action $PickAboveC$ in every high-level executable situation.
	\qed
\end{example}

\subsection{Restrictions on Refinement Mapping}
\looseness=-1
We say a refinement mapping $m$ is \textit{flat}, iff the following two conditions hold: (1) for every high-level action $\act$, $m(\act) \doteq \pi \vec{x}. \psi(\vec{x})?; A(\vec{x}, \vec{c})$; and (2) for every high-level predicate fluent $P$, $m(P) \doteq \exists \vec{x}. \phi(\vec{x})$. 
Given a flat refinement mapping $m$, we use $\Phi(m)$ for the set $\{\phi(\vec{x}) \mid m(P) \doteq \exists \vec{x}. \phi(\vec{x}) \text{ for a high-level predicate fluent } P \text{ or } m(f) \doteq \\  \countOper \vec{x}. \phi(\vec{x}) \text{ for a high-level functional fluent } f \}$.

\looseness=-1
We say a flat refinement mapping $m$ is \textit{complete}, iff the following two conditions hold: (1) for every high-level action $\act$ and every high-level predicate fluent $P$, $\act$ is $P$-enabling, $P$-disabling, or $P$-invariant; and (2) for every high-level action $\act$ and every high-level functional fluent $f$, $\act$ is $f$-incremental, $f$-decremental, or $f$-invariant.

\looseness=-1
We say a low-level situation $s_l$ is $m$-\textit{reachable}, iff $M^l, v[s/s_l] \models \Do(\pi^*_m, S_0, s)$.
We say two low-level situations $s_1$ and $s_2$ of two low-level structures $M_1$ and $M_2$ are \textit{in the same abstract state}, iff the following two conditions hold: (1) for every high-level predicate fluent $P$ and every  variable assignment $v$, we have $M_1, v[s/s_1] \models m(P)[s]$ iff $M_2, v[s/s_2] \models m(P)[s]$; and (2) for every high-level predicate fluent $f$ and every  variable assignment $v$,  $M_1, v[s/s_1] \models m(f)[s] = y$ iff $M_2, v[s/s_2] \models m(f)[s] = y$.

\looseness=-1
We say a refinement mapping $m$ is \textit{executability-preserving}, iff for every two low-level $m$-reachable situations $s_1$ and $s_2$ in the same abstract state, we have $M_1, v[s/s_1] \models \exists s'. \Do(m(\act), s, s')$ iff $M_2, v[s/s_2] \models \exists s'. \Do(m(\act), s, s')$ for every high-level action $\act$.

\begin{theorem} \label{thm:abstraction}
	Suppose that the refinement mapping $m$ is flat, complete and executability-preserving.
	Let $\BAT^h = \Sigma \cup \BAT^h_{S_0} \cup \BAT^h_{ap} \cup \BAT^h_{ss} \cup \BAT^h_{una}$ where 	
	
	\begin{itemize}
		\item $\highBAT_{S_0} \equiv \forget(\phi \land \lowBAT_{fma} \land \Psi_m, \predSet^l \cup \set{\codingSym})[S_0]$ where $\phi[S_0] \equiv \lowBAT_{S_0}$;
				
		\item $\highBAT_{ap}$ contains the set of sentences: for every high-level action $\act$, $\Poss(\act, s) \equiv \forget((\exists \vec{x}. \psi(\vec{x}) \land \Pi_A(\vec{x}, \vec{c})) \land \BAT^{l-}_{con} \land \lowBAT_{fma} \land \Psi_m, \predSet^l \cup \set{\codingSym})[s]$
		where $m(\act) \doteq \pi \vec{x}. \psi(\vec{x})?; A(\vec{x}, \vec{c})$; 
				
		\item $\BAT^h_{ss}$ contains the set of sentences:
		\begin{itemize}
			\item for every high-level predicate fluent $P$, 
			\begin{center}
				$P(do(a, s)) \equiv [\bigvee \limits_{i = 1}^{m} a = \act^+_i] \lor [P(s) \land \bigwedge \limits_{j = 1}^{n} a \neq \act^-_j]$
			\end{center}
			where each $\act^+_i$ is $P$-enabling relative to $m$ and each $\act^-_j$ is $P$-disabling.
			
			\item for every high-level functional fluent $f$, 
			\begin{center}
				$f(do(a, s)) = y \equiv [y = f(s) + 1 \land (\bigvee \limits_{i = 1}^{m} a = \act^+_i)] \lor$ \\
				$\hspace*{27.5mm}[y = f(s) - 1 \land (\bigvee \limits_{i = 1}^{n} a = \act^-_i)] \lor$ \\
				$\hspace*{20mm} [y = f(s) \land \bigwedge \limits_{i = 1}^{m} a \neq \act^+_i \land \bigwedge \limits_{i = 1}^{n} a \neq \act^-_i]$
			\end{center}
			where each $\act^+_i$ is $f$-incremental relative to $m$ and each $\act^-_i$ is $f$-decremental.
		\end{itemize}	
	\end{itemize}
	
	Then, $\BAT^h$ is a complete and sound abstraction of $\BAT^l$ relative to $m$.
\end{theorem}

\looseness=-1
By Theorem \ref{thm:abstraction}, we can directly acquire the successor state axioms of a complete and sound abstraction.
The high-level initial KB can be computed via forgetting every low-level predicate fluent symbol and the coding symbol in the low-level initial KB conjoining with the finitely many objects axiom and the mapping formula.
The action precondition axioms can be similarly obtained.

\begin{example} \label{exm:flatCompRefMapSSA}
	\looseness=-1
	The high-level BAT $\BAT^{BW_h}$ contains a predicate fluent $Holding$, a functional fluent $Num$ and two actions: $Putdown$ and $PickAboveC$.
	As mentioned in Example \ref{exm:highActEnabling}, high-level action $Putdown$ is $Holding$-disabling and $Num$-invariant relative to $m$ while $PickAboveC$ is $Holding$-enabling and $Num$-decremental.
	Hence, $m$ is flat and complete.
	In addition, $m$ is also executability-preserving, which will be verified later.
	
	\looseness=-1
	By Theorem \ref{thm:abstraction}, we construct the successor state axiom of the complete and sound abstraction $\BAT^{BW_h}$ of $\BAT^{BW_l}$. 
		
	Since action $Putdown$ is $Holding$-disabling relative to $m$ and $PickAboveC$ is $Holding$-enabling, the successor state axiom for predicate fluent $Holding$ is
	\begin{center}
		$\hspace*{-47.5mm} Holding(do(a, s)) \equiv$ \\
		$\hspace*{10mm} [a = PickAboveC] \lor [Holding(s) \land a \neq Putdown]$.
	\end{center}	
	
	As $PickAboveC$ is $Num$-decremental relative to $m$, the successor state axiom for functional fluent $Num$ is 
	\begin{center}
		$\hspace*{-45mm} Num(do(a, s)) = y \equiv $ \\
		$\hspace*{-4.75mm} [y = Num(s) - 1 \land a = PickAboveC] \lor$ \\
		$\hspace*{10mm} [y = Num(s) \land a \neq PickAboveC]$. \qed
	\end{center}
\end{example}

\looseness=-1
In the following, we define a normal form of low-level situation suppressed formulas $\phi$, namely \textit{propositional} $\exists$-\textit{formula}, and a translation $m^{-1}$ that maps every low-level propositional $\exists$-formula to a high-level formula.
We also identify four conditions on flat refinement mapping under the translation $m^{-1}$ derives the result of forgetting for computing the initial KB and the action precondition axioms.
 
\looseness=-1
Let $\Phi$ be a set of formulas.
We say a formula is a \textit{propositional} $\exists$-\textit{formula over} $\Phi$, iff it is the Boolean combination of $\exists \vec{x}. \phi(\vec{x})$'s where $\phi(\vec{x}) \in \Phi$.

\looseness=-1
We say a flat refinement mapping $m$ is \textit{simply forgettable}, iff $\forget(\phi \land \BAT^{l-}_{con} \land \BAT^{l}_{fma} \land \Psi_m, \predSet^l \cup \set{\codingSym}) \equiv \forget(\phi \land \Psi_m, \predSet^l)$ for every propositional $\exists$-formula $\phi$ over $\Phi(m)$.
Simple forgettability property says that forgetting all low-level predicate fluent symbols and the coding symbol in any propositional $\exists$-formula $\phi$ together with the state constraint $\BAT^{l-}_{con}$, the finitely many object axiom $\BAT^{l}_{fma}$ and the mapping formula $\Psi_m$ can be simplified to forgetting all low-level predicate fluent symbols in the conjunction of $\phi$ and $\Psi_m$.

\begin{proposition} \label{prop:BWSimpForg}
	The refinement mapping $m$ for blocks world defined in Example \ref{exm:refMap} is simply forgettable.
\end{proposition}
 
\looseness=-1
Given a flat refinement mapping $m$ and a low-level propositional $\exists$-formula $\varphi$ over $\Phi(m)$, the formula $m^{-1}(\varphi)$ is obtained from $\varphi$ by the following two steps:
\begin{enumerate}
	\item generate the formula $\psi$ by replacing every occurrence of $\exists \vec{x}. \phi(\vec{x})$ in $\varphi$ by $P$ (resp. $f > 0$) where $m(P) \doteq \exists \vec{x}. \phi(\vec{x})$ (resp. $m(f) \doteq \countOper \vec{x}. \phi(\vec{x})$);
 	
 	\item generate the formula $m^{-1}(\varphi)$ by conjoining $\psi$ with the conjunct $\bigwedge_{f \in \funcSet^h(\psi)} f \geq 0$.
\end{enumerate}

\looseness=-1
We say a flat refinement mapping $m$ is \textit{consistent}, iff $\exists \vec{x}. \phi(\vec{x})$ is consistent for every high-level predicate fluent $P$ and every high-level functional fluent $f$ s.t. $m(P) \doteq \exists \vec{x}. \phi(\vec{x})$ and $m(f) \doteq \countOper \vec{x}. \phi(\vec{x})$.

\looseness=-1
We say a flat refinement mapping $m$ is \textit{syntax-irrelevant}, iff for every two high-level predicate (or functional) fluents $P_1$ (or $f_1$) and $P_2$ (or $f_2$) s.t. $m(P_1) \doteq \exists \vec{x}. \phi_1(\vec{x})$ (or $m(f_1) \doteq \countOper \vec{x}. \phi_1(\vec{x})$) and $m(P_2) \doteq \exists \vec{y}. \phi_2(\vec{y})$ (or $m(f_2) \doteq \countOper \vec{y}. \phi_2(\vec{y})$), $\predSet^l(\phi_1) \cap \predSet^l(\phi_2) = \emptyset$.
Syntax-irrelevant refinement mapping requires that for every two high-level predicate (or functional) fluents, the two sets of low-level predicate fluents that occur in the corresponding two low-level situation-suppressed formulas $\phi_1$ and $\phi_2$ defined in $m$ are disjoint.

\begin{proposition} \label{prop:consSynIrRefMap}
	Let $m$ be a flat, consistent and syntax-irrelevant refinement mapping and $\phi$ a propositional $\exists$-formula over $\Phi(m)$.
	Then, $m^{-1}(\phi) \equiv \forget(\phi \land \Psi_m, \predSet^l)$.
\end{proposition}

\looseness=-1
Proposition \ref{prop:consSynIrRefMap} says that the result of forgetting in a propositional $\exists$-formula conjoining with the mapping formula can be obtained by the translation $m^{-1}$.
However, the low-level initial KB and the executability condition of low-level guarded actions may not be in the form of propositional $\exists$-formulas.
We therefore impose one additional condition on refinement mappings such that the above two formulas are equivalent to propositional $\exists$-formulas under the finitely many objects axiom and the state constraints, respectively.

\looseness=-1
We say a flat refinement mapping $m$ is \textit{propositional} $\exists$-\textit{definable}, iff (1) there is a situation-suppressed propositional $\exists$-formula $\phi$ over $\Phi(m)$ s.t. $\BAT^{l-}_{con} \land \lowBAT_{fma} \models \phi \equiv \varphi$ where $\varphi[S_0] = \lowBAT_{S_0}$; and (2) for every high-level action $\act$ s.t. $m(\act) \doteq \pi \vec{x}. \psi(\vec{x})?; A(\vec{x}, \vec{c})$, there is a situation-suppressed propositional $\exists$-formula $\phi$ over $\Phi(m)$ s.t. $\BAT^{l-}_{con} \land \lowBAT_{fma} \models \phi \equiv \exists \vec{x}. \psi(\vec{x}) \land \Pi_A(\vec{x}, \vec{c})$.

\begin{proposition} \label{prop:execPreRefMap}
	A flat, consistent, syntax-irrelevant and propositional $\exists$-definable refinement mapping is executability-preserving.
\end{proposition}

\looseness=-1
Proposition \ref{prop:execPreRefMap} means that the four properties of refinement mappings: flatness, consistency, syntax-irrelevance and propositional $\exists$-definability together guarantees the executability-preserving property.

\looseness=-1
By putting Theorem \ref{thm:abstraction} and Propositions \ref{prop:forgetUnderTheory}, \ref{prop:consSynIrRefMap} and \ref{prop:execPreRefMap} into together, we obtain a syntactic approach to computing sound and complete abstractions given a flat, complete,  simply forgettable, consistent, syntax-irrelevant and propositional $\exists$-definable refinement mapping.
We illustrate this approach with the following example.

\begin{example} \label{exm:flatSimpRefMapping}	
	\looseness=-1
	Remember that $m(Num) \doteq \countOper x. on^{+}(x, C)$ and $m(Holding) \doteq \exists x. holding(x)$.
	The two formulas $\exists x. on^{+}(x, C)$ and $\exists x. holding(x)$ are consistent, so $m$ is consistent.
	In addition, $on$ and $holding$ are different low-level predicate fluent symbols, hence $m$ is syntax-irrelevant.
	The set $\Phi(m)$ of formulas is $\set{on^+(x, C), holding(x)}$.
	
	\looseness=-1
	Example \ref{exm:BATLIA} presents the sound and complete abstraction $\BAT^{BW_h}$ of $\BAT^{BW_l}$.	
	We will show the construction of $\BAT^{BW_h}$ according to Theorem \ref{thm:abstraction} and Propositions \ref{prop:forgetUnderTheory}, \ref{prop:BWSimpForg} - \ref{prop:execPreRefMap}.	
	The successor state axioms was shown in Example \ref{exm:flatCompRefMapSSA}.
	
	\looseness=-1
	We first construct the high-level initial KB $\BAT^{BW_h}_{S_0}$.
	The low-level initial KB $\BAT^{BW_l}_{S_0}$ is $(\varphi \land \BAT^{BW_l}_{con})[S_0]$ where $\varphi: \neg \exists x. holding(x) \land \exists x. on^{+}(x, C)$ is a propositional $\exists$-formula over $\Phi(m)$.
	Clearly, $\BAT^{BW_l-}_{con} \land \BAT^{BW_l}_{fma} \models \varphi \equiv (\varphi \land \BAT^{BW_l}_{con})$.
	The formula $m^{-1}(\varphi) = \neg Holding \land Num > 0 \land Num \geq 0$.
	Hence, we get that $\BAT^{BW_h}_{S_0} \equiv \neg Holding(S_0) \land Num(S_0) > 0$.
	
	\looseness=-1
	We then construct the action precondition axioms $\BAT^{BW_h}_{ap}$.
	For high-level action $Putdown$, $m(Putdown) \doteq \pi x. putdown(x)$ and $\Poss(putdown(x), s) \equiv holding(x, s)$.
	The formula $\exists x. holding(x)$ is a propositional $\exists$-formula over $\Phi(m)$.
	Since $m^{-1}(\exists x. holding(x)) = Holding$, we obtain that $\Poss(Putdown, s) \equiv Holding(s)$.
	
	\looseness=-1
	For high-level action $PickAboveC$, $m(PickAboveC) \doteq (\pi x, y) on^+(x, C)?; unstack(x, y)$ and $\Poss(unstack(x, y), s) \equiv on(x, y, s) \land \neg (\exists z) on(z, x, s) \land \neg (\exists z) holding(z, s)$.
	The executability condition $\phi$ of guaraded action $(\pi x, y) on^+(x, C)?; unstack(x, y)$ is $\exists x, y. on^+(x, C) \land on(x, y) \land \neg (\exists z) on(z, x) \land \neg (\exists z) holding(z)$.
	Clearly, $\phi$ is not a propositional $\exists$-formula over $\Phi(m)$.
	By the transitive closure property of $on^+(x, y)$, the finitely many objects axiom and the first state constraint $\neg on^+(x, x)$ of $\BAT^{BW_l-}_{con}$, we get that 
	$\BAT^{BW_l-}_{con} \land \BAT^{BW_l}_{fma} \models (\exists x) on^+(x, C) \equiv (\exists x, y. on^+(x, C) \land on(x, y) \land \neg (\exists z) on(z, x))$.
	It follows that $\BAT^{BW_l-}_{con} \land \BAT^{BW_l}_{fma}  \models [(\exists x) on^+(x, C) \land \neg (\exists x) holding(x)] \equiv \phi$.
	In addition, $m^{-1}((\exists x) on^+(x, C) \land \neg (\exists x) holding(x)) = Num > 0 \land \neg Holding \land Num \geq 0$.
	Hence, $\Poss(PickAboveC, s) \equiv Holding(s) \land Num(s) > 0$.
	\qed
\end{example}

\section{Related Work}
\looseness=-1
In generalized planning, \citet{BonFG2019} utilize SAT solvers and description logics to learn an abstract qualitative numerical problem (QNP) \cite{SriZIG2011} from a finite subset of problem instances.
However, the learned QNP may not be a sound abstraction for the whole set of original problem instances.
Later, \citet{BonFE2019} studied the necessary condition of problem instances on which the abstraction is guaranteed to be sound.
In addition, QNP under deterministic semantics (that is, each integer variable $x$ can be decremented by $1$ if $x > 0$ or incremented by $1$) \cite{SriZGAR2015} and generalized linear integer numeric planning problem \cite{LinCFGLS2022} can be formalized in linear integer situation calculus.

\looseness=-1
The abstraction-refinement framework is a popular paradigm for verifying properties against multi-agent systems.
It first initializes an abstract model based on 3-valued semantics.
If the property is true (or false) in the abstract model, then it is also true (or false) in the concrete model.
Otherwise, the abstract model is too coarse to decide the satisfaction of the property.
In this case, the framework will search some failure states to refine the abstraction in an attempt to return a definite answer.
Various abstraction and refinement strategies are proposed.
(1) \citet{BelFF2023} constructs an initial abstract model based on the common knowledge relationships among agents, and then refines the abstract model by splitting the failure state $s$ according to the incoming transitions into $s$.
(2) Predicate abstraction \cite{LomAM2016,BelFL2016} first selects key predicates to represent the abstract model and then updates the predicate set according to failure states so as to refine the model.
(3) \citet{LomuscioM2014} combines state abstraction and action abstraction to constructs an abstract model, and then refines the model by splitting those states where the truth value of some subformula of the property is unknown.
However, for alternating-time temporal logic under imperfect information and perfect recall \cite{BelFF2023} and alternating-time temporal logic with knowledge \cite{LomAM2016}, the abstraction-refinement framework cannot guarantee that a complete abstract model is obtained since finding failure states in these two logics is undecidable.
The key difference between the above works and our paper is that they focuses on abstraction on models while we investigate abstraction on theories.

\looseness=-1
An abstraction framework based on the situation calculus and ConGolog programming language was proposed in \cite{BanihashmemiGL2017}.
Later, this abstraction framework was extended to cope with online execution with sensing actions \cite{BanihashemiGL2018}, probabilistic domains \cite{Bel2020,HofB2023}, generalized planning \cite{CuiLL2021}, non-deterministic actions \cite{BanihashemiGL2023}, and multi-agent games \cite{LesperanceGRK2024}.
\citet{CuiKL2023} explored the automatic verification of sound abstractions for generalized planning. 
However, none of these aforementioned works investigates the computation problem of abstractions.
\citet{LuoLL2020} showed that sound and complete abstractions can be characterized theoretically via the notion of forgetting.
However, forgetting is a computationally difficult problem.
Hence, the above method is not a practical approach.


\section{Conclusions and Future Work}
\looseness=-1
In this paper, we develop a syntactic approach to computing complete and sound abstraction in the situation calculus.
To this end, we first present a variant of situation calculus, namely linear integer situation calculus, using LIA as the underlying logic.
We then migrate \citet{BanihashmemiGL2017}'s abstraction framework to one from linear integer situation calculus to extended situation calculus with the following modifications:
(1) The high-level BAT is based on linear integer situation calculus allowing functional fluents.
(2) The refinement mapping maps every high-level action to a closed low-level Golog program, maps every high-level predicate fluent to a closed low-level formula, and maps every high-level functional fluent to a closed low-level counting term.
In addition, we impose six restrictions (flatness, completeness, simple forgettability, consistency, syntax-irrelevancy and propositional $\exists$-definability) on refinement mapping.
We also define a translation $m^{-1}$ from low-level propositional $\exists$-formulas to high-level formulas.
Under refinement mappings with the above restrictions, initial KB, action precondition axioms and successor state axioms of high-level BAT can be syntactically constructed from the corresponding low-level counterparts via the translation $m^{-1}$.

\looseness=-1
This paper opens several avenues for further work.
(1) While this paper provides theoretical results on abstraction, it is interesting to implement the syntactic approach so as to automatically produce high-level action theories.
(2) This paper assumes that the restricted refinement mapping is provided.
We will explore how to generate suitable refinement mappings from low-level action theories.
(3) The four restrictions: completeness, simple-forgettability, consistency and propositional $\exists$-definability are not syntactical.
One possible direction is to develop an automatic method to verify if a refinement mapping satisfies the above four restrictions. 
(4) We plan to loosen the restrictions in this paper so as to devise a more general syntactic approach to suit more complex planning domains.

\section{Acknowledgments}
\looseness=-1
We are grateful to Yongmei Liu, Hong Dong and Shiguang Feng for their helpful discussions on the paper.
This paper was supported by National Natural Science Foundation of China (Nos. 62276114, 62077028, 62206055 and 62406108), Guangdong Basic and Applied Basic Research Foundation (Nos. 2023B1515120064, 2024A1515011762), Guangdong-Macao Advanced Intelligent Computing Joint Laboratory (No. 2020B1212030003), Science and Technology Planning Project of Guangzhou (No. 2025A03J3565), Research Foundation of Education Bureau of Hunan Province of China (No. 23B0471), the Fundamental Research Funds for the Central Universities, JNU (No. 21623202).

\bibliography{AAAI-2025-1}

\newpage
\appendix
\section{Supplemental Proofs}

{\noindent \bf Proof of Proposition \ref{prop:forgetUnderTheory}}
\begin{proof}
	Since $\theory \models \phi \equiv \psi$, $\theory \cup \set{\phi} \equiv \theory \cup \set{\psi}$.
	Hence, $\forget(\theory \cup \set{\phi}, \predFuncSet) \equiv \forget(\theory \cup \set{\psi}, \predFuncSet)$.
\end{proof}

{\noindent \bf Proof of Proposition \ref{prop:isomorphic}}
\begin{proof}
	By induction on the structure of $\phi$.
\end{proof}

{\noindent \bf Proof of Proposition \ref{prop:isomorphicForget}}

\begin{proof}
	Suppose that $M^l, v[s/s_l] \models \phi[s] \land \lowBAT_{fma}$.
	We can construct a model $M$ s.t. $M \models \phi \land \lowBAT_{fma} \land \Psi_m$ and $M \models P^l(\vec{x})$ iff $M^l, v[s/s_l] \models P^l[s]$ for every low-level predicate fluent $P^l$ and $(\codingSym(x))^M = (\codingSym(x))^{M^l}$ for every object $x$.
	It follows that $M \models P^h$ iff $M \models m(P^h)$ for every high-level predicate fluent $P^h$, and similarly for every high-level functional fluent $f^h$.
	Since $s_h \bisimilar_{m}^{M^h, M^l} s_l$, $M^l, v[s/s_l] \models m(P^h)[s]$ iff $M^h, v[s/s_h] \models P^h[s]$ for every high-level predicate fluent $P^h$, and similarly for every high-level functional fluent $f^h$.
	Thus, $M \models P^h$ iff $M^h, v[s/s_h] \models P^h[s]$ for every high-level predicate fluent $P^h$, and similarly for every high-level functional fluent $f^h$.
	We get that $M^h, v[s/s_h] \models \forget(\phi \land \lowBAT_{fma} \land \Psi_m, \predSet^l \cup \set{\codingSym})[s]$.
\end{proof}

\begin{lemma} \label{lem:m-bisim}
	Let $M^h$ be a high-level model and $M^l$ a low-level model s.t. $M^h \sim_m M^l$.
	Let $\vec{\act}$ be a high-level action sequence.
	Let $s_h$ be a high-level situation $do(\vec{\act}, S_0)^{M^h}$ s.t. $M^h \models \Exec(do(\vec{\act}, S_0))$ and $s_l$ a low-level situation s.t. $M^l, v[s/s_l] \models \Do(m(\vec{\act}), S_0, s)$.
	Then, for every $m$-bisimulation $B$ s.t. $(S_0^{M^h}, S_0^{M^l}) \in B$, we have $(s_h, s_l) \in B$.	
\end{lemma}
\begin{proof}
	By induction of the length of $\vec{\act}$.
	Let $B$ be a $m$-bisimulation over $\sitDom^{M^h} \times \sitDom^{M^l}$ s.t. $(S_0^{M^h}, S_0^{M^l}) \in B$.
	
	\textbf{Base case} ($s_h = S_0^{M^h}$ and $s_l = S_0^{M^l}$): $(s_h, s_l) \in B$ due to $(S_0^{M^h}, S_0^{M^l}) \in B$.
	
	\textbf{Inductive step}: Assume that $s_h$ is a high-level situation $do(\vec{\act}; \beta, S_0)^{M^h}$ s.t. $M^h \models \Exec(do(\vec{\act}; \beta, S_0))$  and $s_l$ is a low-level situation s.t. $M^l, v[s/s_l] \models \Do(m(\vec{\act}; \beta), S_0, s)$.
	Thus, there is a high-level situation $s^*_h$ s.t. $s^*_h = do(\vec{\act}, S_0)^{M^h}$ and $s_h = do(\beta, s^*_h)^{M^h}$.
	It follows that $M^h, v[s^*/s^*_h] \models \Poss(\beta, s^*)$.
	There is a low-level situation $s^*_l$ s.t. $M^l, v[s^*/s^*_l, s/s_l] \models \Do(m(\vec{\act}), s^*, S_0) \land \Do(m(\beta), s^*, s)$.
	By the inductive hypothesis, $(s^*_h, s^*_l) \in B$.
	By the forth condition of $m$-bisimulation, $(s_h, s_l) \in B$.
\end{proof}

\begin{lemma} \label{lem:soundAbs}
	Let $m$ be a refinement mapping from $\highBAT$ to $\lowBAT$.	
	Suppose that $M^h \models \highBAT$ and $M^l \models \lowBAT$ s.t. the following conditions hold
	\begin{enumerate}
		\item $S_0^{M^h} \bisimilar_m^{M^h, M^l} S_0^{M^l}$;
		
		\item $\lowModel \models \forall s. \Do(\pi^*_m, S_0, s) \implies \bigwedge_{\act \in \ \actSet^h} (m(\Pi_{\act})[s] \equiv \exists s' \Do(m(\act), s, s'))$ where $\Poss(\act, s) \equiv \Pi_{\act}[s]$;
		
		\item $\lowModel \models \forall s. \Do(\pi^*_m, S_0, s) \implies \\
		\blank \bigwedge_{\act \in \ \actSet^h} \forall s' \{ \Do(m(\act), s, s') \implies \\
		\blank \blank \bigwedge_{P \in \predSet^h} (m(P)[s'] \equiv m(\tilde{\Phi}_{P, \act})[s]) \land \\
		\blank \blank \bigwedge_{f \in \funcSet^h} \forall y(m(f)[s'] = y \equiv m(\tilde{\Phi}_{f, \alpha})(y)[s]) \}$
		where $P(do(a, s)) \equiv \Phi_P(a, s)$, $f(do(a, s)) = y \equiv \Phi_f(y, a, s)$, $\tilde{\Phi}_{P, \act}$ is the simplified formula of $\Phi_P(a, s)$ and $\tilde{\Phi}_{f, \alpha}(y)$ is the simplified formula of $\Phi_f(y, a, s)$.
	\end{enumerate}
	Then, $M^h \bisimilar_m M^l$.
\end{lemma}
\begin{proof}
	Assume that $M^h$ and $M^l$ satisfy the above three conditions.
	We verify that $M^h$ is $m$-bisimilar to $M^l$.	
	We construct the relation $B$ over $\sitDom^{M^h} \times \sitDom^{M^l}$ as follows: $(s_h, s_l) \in B$ iff there is a high-level action sequence $\vec{\act}$ s.t. $M^l, v[s/s_l] \models \Do(m(\vec{\act}), S_0, s)$ and $s_h = do(\vec{\act}, S_0)^{M^h}$.
	Clearly, $(S_0^{M^h}, S_0^{M^l}) \in B$ and every situation $s^l$ s.t. $(s_h, s_l) \in B$ is $m$-reachable.
	We verify that the atom, forth and back conditions of $m$-bisimulation by induction on the length of situation $s_h$.
	
	\textbf{Base case} ($s_h = S_0^{M^h}$):
	In the following, we verify the atom, forth and back conditions of $m$-bisimulation.
	\begin{description}
		\item[atom] It directly derive from Condition 1 that $S_0^{M^h}$ and $S_0^{M^l}$ are $m$-isomorphic.
		
		\item[forth] Let $\beta$ be a high-level action where the action precondition axiom for $\beta$ is $\Poss(\beta, s) \equiv \Pi_\beta[s]$.
		Let $s'_l$ be a low-level situation s.t. $M^l, v[s'/s'_l] \models \Do(m(\beta), S_0, s')$.
		We let $s'_h$ be $do(\beta, S_0)^{M^h}$.
		By the construction of relation $B$, $(s'_h, s'_l) \in B$.
		It remains to show that $M^h \models \Poss(\beta, S_0)$.
		Since $M^l, v[s'/s'_l] \models \Do(m(\beta), S_0, s')$, $M^l \models \exists s'. \Do(m(\beta), S_0, s')$.
		By Condition 2, we get that $M^l \models m(\Pi_\beta)[S_0]$.
		By Proposition \ref{prop:isomorphic}, we get that $M^h \models \Pi_\beta[S_0]$.
		Hence, $M^h \models \Poss(\beta, S_0)$.
		
		\item[back] Let $\beta$ be a high-level action where the action precondition axiom for $\beta$ is $\Poss(\beta, s) \equiv \Pi_\beta[s]$.		
		Let $s'_h$ be a high-level situation s.t. $M^h, v[s'/s'_h] \models \Poss(\beta, S_0) \land s' = do(\beta, S_0)$.
		It follows that $M^h \models \Pi_\beta[S_0]$.
		By Proposition \ref{prop:isomorphic}, $M^l \models m(\Pi_\beta)[S_0]$.
		By Condition 2, we get that $M^l \models \exists s'. \Do(m(\beta), S_0, s')$.
		Thus, there is a low-level situation $s'_l$ s.t. $M^l, v[s'/s'_l] \models \Do(m(\beta), S_0, s')$.
		By the construction of relation $B$, $(s'_h, s'_l) \in B$.
	\end{description}
	
	\textbf{Inductive step}: Suppose that $(s'_h, s'_l) \in B$ where $s_h = do(\vec{\act}, S_0)^{M^h}$ and $s'_h = do(\beta, s_h)^{M^h}$ with the action precondition axiom for $\beta$ is $\Poss(\beta, s) \equiv \Pi_\beta[s]$.
	We here only verify that $s'_h$ and $s'_l$ are $m$-isomorphic.
	Similarly to the base case, the forth and back conditions can be proved.
	By the construction of relation $B$ and the semantics of Golog programs, there is a situations $s_l$ in $M^l$ s.t. $M^l, v[s/s_l, s'/s'_l] \models \Do(m(\vec{\act}), S_0, s) \land \Do(m(\beta), s, s')$.
	It is easily verified that $(s_h, s_l) \in B$.	
	
	We first verify that for every high-level predicate fluent $P$ and every variable assignment $v$, $M^h, v[s'/s'_h] \models P(s')$ iff $M^l, v[s'/s'_l] \models m(P)[s']$.
	
	\begin{itemize}
		\item ($\Rightarrow$): 
		Suppose that $M^h, v[s'/s'_h] \models P(s')$.
		By the successor state axiom for $P$, $M^h, v[s/s_h] \models \tilde{\Phi}_{P, \act}[s]$.
		By Proposition \ref{prop:isomorphic}, $M^l, v[s/s_l] \models m(\tilde{\Phi}_{P, \act})[s]$.
		By Condition 3, $M^l, v[s'/s'_l] \models m(P)[s']$.
		
		\item ($\Leftarrow$): Suppose that $M^l, v[s'/s'_l] \models P(s')$.
		By Condition 3, $M^l, v[s/s_l] \models m(\tilde{\Phi}_{P, \act})[s]$.
		By Proposition \ref{prop:isomorphic}, $M^h, v[s/s_h] \models \tilde{\Phi}_{P, \act}[s]$.
		By the successor state axiom for $P$, $M^h, v[s'/s'_h] \models P(s')$.
	\end{itemize}
		
	We now verify that for every high-level functional fluent $f$ and every variable assignment $v$, $M^h, v[s'/s'_h] \models f(s') = y$ iff $M^l, v[s'/s'_l] \models m(f)[s'] = y$.
	\begin{itemize}
		\item ($\Rightarrow$): Suppose that $M^h, v[s'/s'_h] \models f(s') = y$.
		By the successor state axiom for $f$, $M^h, v[s/s_h] \models \tilde{\Phi}_{f, \beta}(y)[s]$.
		By Proposition \ref{prop:isomorphic}, $M^l, v[s/s_l] \models m(\tilde{\Phi}_{f, \beta})(y)[s]$.
		By Condition 3, $M^l, v[s'/s'_l] \models m(f)[s'] = y$.
		
		\item ($\Leftarrow$): Suppose that $M^l, v[s'/s'_l] \models m(f)[s'] = y$.
		By Condition 3, $M^l, v[s/s_l] \models m(\tilde{\Phi}_{f, \beta})(y)[s]$.
		By Proposition \ref{prop:isomorphic}, $M^h, v[s/s_h] \models \tilde{\Phi}_{f, \beta}(y)[s]$.
		By the successor state axiom for $f$, $M^h, v[s'/s'_h] \models f(s') = y$.		
	\end{itemize}
\end{proof}

\begin{lemma} \label{lem:relSatLowBAT}
	Let $M$ be a model of $\lowBAT_{S_0} \land \lowBAT_{fma}$.
	There is a model $M^l$ of $\lowBAT$ s.t. $M \models P(\vec{x}, S_0)$ iff $M^l \models P(\vec{x}, S_0)$ for every low-level predicate fluent $P$. 
\end{lemma}
\begin{proof}
	Since $\lowBAT_{una}$ involves no any symbol of $M$, we can extend $M$ s.t. $M \models \lowBAT_{S_0} \land \lowBAT_{fma} \land \lowBAT_{una}$.
	According to the proof of the relative satisfiability of basic action theories (Theorem 1 in \cite{PirriR1999}), we can extend $M$ to a model $M^l$ of $\lowBAT$.	
\end{proof}

\begin{lemma} \label{lem:relSatHighBAT}
	Let $M$ be a model of $\highBAT_{S_0}$.
	There is a model $M^h$ of $\highBAT$ s.t. $M \models P(S_0)$ iff $M^h \models P(S_0)$ for every high-level predicate fluent $P$, and similarly for every high-level functional fluent $f$. 
\end{lemma}
\begin{proof}
	The proof is similar to that of Lemma \ref{lem:relSatLowBAT} except that we do not contain the finitely
	many objects axiom and involve functional fluent.
\end{proof}

{\noindent \bf Proof of Theorem \ref{thm:soundAbs}}
\begin{proof}	
	($\Rightarrow$):
	
	\begin{enumerate}
		\item Let $M^l_{S_0}$ be a model of $\lowBAT_{S_0} \land \lowBAT_{fma}$.
		By Lemma \ref{lem:relSatLowBAT}, $M^l_{S_0}$ can be extended to a model $M^l$ satisfy all axioms of $\lowBAT$.
		Since $\highBAT$ is a sound abstraction of $\lowBAT$ relative to $m$, there is a model $M^h$ of $\highBAT$ s.t. $M^h \sim_m M^l$.
		So $S_0^{M^h}$ and $S_0^{M^l}$ are $m$-isomorphic.
		In addition, $M^h \models \highBAT_{S_0}$.
		By Proposition \ref{prop:isomorphic}, $M^l \models m(\phi)[S_0]$ where $\highBAT_{S_0} \equiv \phi[S_0]$.
		Since $M^l$ and $M^l_{S_0}$ agree on the interpretation of every fluent over the initial situation $S_0$, $M^l_{S_0} \models m(\phi)[S_0]$.
		
		\item Let $M^l$ be a model of $\lowBAT$.
		Let $s_l$ be a low-level situation s.t. $M^l, v[s/s_l] \models \Do(\pi^*_m, S_0, s)$.
		There is a high-level action sequence $\vec{\act}$ s.t. $M^l, v[s/s_l] \models \Do(\vec{\act}, S_0, s)$.
		Since $\highBAT$ is a sound abstraction of $\lowBAT$ relative to $m$, there is a model $M^h$ of $\highBAT$ s.t. $M^h \sim_m M^l$.
		Let $s_h$ be a high-level situation $do(\vec{\act}, S_0)^{M^h}$.
		Let $B$ be the $m$-bisimulation s.t. $(S_0^{M^h}, S_0^{M^l}) \in B$.
		By Lemma \ref{lem:m-bisim}, we get that $(s_h, s_l) \in B$.
		It follows that $s_h$ and $s_l$ are $m$-isomorphic.
		Let $\beta$ be a high-level action where $\Poss(\beta, s) \equiv \Pi_\beta[s]$.
		
		Suppose that $M^l, v[s/s_l] \models m(\Pi_\beta)[s]$.
		By Proposition \ref{prop:isomorphic}, $M^h, v[s/s_h] \models \Pi_\beta[s]$.
		Hence, $M^h, v[s/s_h] \models \Poss(\beta, s)$.
		Thus, there is a high-level situation $s'_h$ s.t. $s'_h = do(\beta, s)^{M^h}$.		
		By the back condition of $m$-bisimulation, there is a low-level situation $s_l'$ s.t. $M^l, v[s/s_l, s'/s_l'] \models \Do(m(\beta), s, s')$.
		Hence, $M^l, v[s/s_l] \models \exists s'. \Do(m(\beta), s, s')$.
		
		Suppose that $M^l, v[s/s_l] \models \exists s'. \Do(m(\beta), s, s')$.
		Let $s'_l$ be a low-level situation s.t. $M^l, v[s/s_l, s'/s_l'] \models \Do(m(\beta), s, s')$.
		By the forth condition of $m$-bisimulation, there is a high-level situation $s'_h$ s.t. $M^h, v[s/s_h, s'/s'_h] \models \Poss(\beta, s) \land s' = do(\beta, s)$.
		Hence, $M^h, v[s/s_h] \models \Pi_\beta[s]$.
		By Proposition \ref{prop:isomorphic}, $M^l, v[s/s_l] \models m(\Pi_\beta)[s]$.
				
		\item We here only verify the case for high-level predicate fluent $P$.
		The case for high-level functional fluent $f$ can be similarly proved.
		
		Let $\beta$ be a high-level action and the successor state axiom for $P$ is $P(do(a, s)) \equiv \Phi_P(a, s)$.
		Thus, $\highBAT_{ss} \models P(do(\beta, s)) \equiv \tilde{\Phi}_{P, \beta}[s]$.
		Let $M^l$ be a model of $\lowBAT$.
		Let $s_l$ and $s'_l$ be two low-level situations s.t. $M^l, v[s/s_l, s'/s'_l] \models \Do(\pi^*_m, S_0, s) \land \Do(m(\act), s, s')$.
		There is a high-level action sequence $\vec{\act}$ s.t. $M^l, v[s/s_l] \models \Do(\vec{\act}, S_0, s)$.
		Since $\highBAT$ is a sound abstraction of $\lowBAT$ relative to $m$, there is a model $M^h$ of $\highBAT$ s.t. $M^h \sim_m M^l$.
		Let $s_h$ and $s'_h$ be two high-level situations $do(\vec{\act}, S_0)^{M^h}$ and $do(\vec{\act}; \beta, S_0)^{M^h}$, respectively.
		Let $B$ be an $m$-bisimulation s.t. $(S_0^{M^h}, S_0^{M^l}) \in B$.
		By Lemma \ref{lem:m-bisim}, we get that $(s_h, s_l) \in B$.
		It follows that $s_h$ and $s_l$ are $m$-isomorphic.
		Similarly, $s'_h$ and $s'_l$ are also $m$-isomorphic.
		
		Suppose that $M^l, v[s'/s'_l] \models m(P)[s']$.
		By Proposition \ref{prop:isomorphic}, $M^h, v[s'/s'_h] \models P(s')$.
		According to the successor state axiom for $P$, $M^h, v[s/s_h] \models \tilde{\Phi}_{P, \beta}[s]$.
		By Proposition \ref{prop:isomorphic}, $M^l, v[s/s_l] \models m(\tilde{\Phi}_{P, \beta})[s]$.
		
		Suppose that $M^l, v[s'/s'_l] \models m(\tilde{\Phi}_{P, \beta})[s]$.
		By Proposition \ref{prop:isomorphic}, $M^h, v[s/s_h] \models \tilde{\Phi}_{P, \beta}[s]$.
		According to the successor state axiom for $P$, $M^h, v[s'/s'_h] \models P(s')$.
		By Proposition \ref{prop:isomorphic}, $M^l, v[s'/s'_l] \models m(P)[s']$.
		
	\end{enumerate}
	
	($\Leftarrow$): Let $M^l$ be a model of $\lowBAT$.
	We first construct the interpretations of high-level predicate and functional fluent of model $M^h$ at the initial situation $S_0$ as follows:
	\begin{enumerate}		
		\item for each high-level predicate fluent $P \in \predSet^h$, $M^h \models P(S_0))$ iff $M^l \models m(P)[S_0]$.
	
		\item for each high-level functional fluent $f \in \funcSet^h$, we let $f(S_0)^{M^h} = (m(f)[S_0])^{M^l}$.
	\end{enumerate}	

	By the construction, $S_0^{M^h} \bisimilar_m^{M^h, M^l} S_0^{M^l}$.
	By Condition 1 and Proposition \ref{prop:isomorphic}, we get that $M^h \models \initKB^h$.
	By Lemma \ref{lem:relSatHighBAT}, we can extend $M^h$ s.t. $M^h \models \highBAT$.
	Conditions 2 and 3 of Theorem \ref{thm:soundAbs} implies that Conditions 2 and 3 of Lemma \ref{lem:soundAbs} holds.
	By Lemma \ref{lem:soundAbs}, $M^h \sim_m M^l$.
\end{proof}

\begin{lemma} \label{lem:compSoundAbs3}
	Suppose that $\highBAT$ is a sound abstraction of $\lowBAT$ relative to $m$.
	Let $\BAT^{h*}$ be the high-level BAT $\highBAT \setminus \set{\highBAT_{S_0}}$. 
	Let $M^l$ be a model of $\lowBAT$ and $M^h$ a model of $\BAT^{h*}$ s.t. $S_0^{M^h} \bisimilar_m^{M^h, M^l} S_0^{M^l}$.
	Then, $M^h \sim_m M^l$ and $M^h \models \highBAT_{S_0}$.
\end{lemma}
\begin{proof}
	Since $\highBAT$ is a sound abstraction of $\lowBAT$ relative to $m$ and $M^l$ is a model of $\lowBAT$, by Theorem \ref{thm:soundAbs}, Conditions 2 and 3 of Lemma \ref{lem:soundAbs} holds.
	In addition, $S_0^{M^h} \bisimilar_m^{M^h, M^l} S_0^{M^l}$.
	Hence, $M^h \sim_m M^l$.
	
	Since $\highBAT$ is a sound abstraction of $\lowBAT$ relative to $m$, there is a model $N^h$ of $\BAT^h$ s.t. $N^h \bisimilar_m M^l$.
	Hence, $S_0^{M^h} \bisimilar_{m}^{M^h, M^l} S_0^{M^l}$ and $S_0^{N^h} \bisimilar_{m}^{N^h, M^l} S_0^{M^l}$.
	By Proposition \ref{prop:isomorphic}, $M^h \models \highBAT_{S_0}$. 
\end{proof}

{\noindent \bf Proof of Theorem \ref{thm:compSoundAbs}}
\begin{proof}	
	($\Rightarrow$): Assume that $\highBAT$ is a complete abstraction of $\lowBAT$ relative to $m$.
	
	We first verify the only-if direction of forgetting.	
	Let $M$ be a model of $\lowBAT_{S_0} \land \lowBAT_{fma} \land \Psi_m[S_0]$ and $M'$ a structure s.t. $M \identical_{\predSet^l \cup \set{\codingSym}} M'$.
	It follows that for every high-level predicate fluent $P$, $M' \models P(S_0)$ iff $M \models m(P)[S_0]$, and similarly for every high-level functional fluent $f$.
	
	By Lemma \ref{lem:relSatLowBAT}, from $M$, we construct the low-level model $M^l$ satisfying all axioms of $\lowBAT$.

	By Lemma \ref{lem:relSatHighBAT}, we construct the high-level model $M^h$ from $M'$.
	Remember that in the definition of LIBAT, for each functional fluent $f$, there is a successor state axiom of the form $f(do(a, s)) = y \equiv \Phi_f(y, a, s)$ where $\Phi_f(y, a, s)$ is defined as:	
	\begin{center}
		$[\bigvee\limits_{i = 1}^{m} (y = t_{i}[s] \land a = \act^+_i \land \gamma^+_{i}[s])] \lor$ \\ 
		$[y = f(s) \land \neg (\bigvee\limits_{j = 1}^{n} (a = \act^-_j \land \gamma^-_{j}[s]))].$
	\end{center}	
	It is easily verified that $\Phi_f(y, a, s)$ satisfies functional fluent consistency property.
	Similarly to the model $M$, $M'$ can be extended to a model $M^h$ satisfying all axioms of $\highBAT \setminus \set{\highBAT_{S_0}}$.

	It is easily verified that for every high-level predicate fluent $P$, $M^l \models P(S_0)$ iff $M^h \models m(P)[S_0]$, and similarly for every high-level functional fluent $f$.
	$\highBAT$ is a sound abstraction of $\lowBAT$ relative to $m$.
	By Lemma \ref{lem:compSoundAbs3}, $M^h \sim_m M^l$ and $M^h \models \highBAT_{S_0}$.
	It follows that $M' \models \highBAT_{S_0}$.

	We now verify the if direction of forgetting.
	Let $M'$ be a model of $\highBAT_{S_0}$.
	By Lemma \ref{lem:relSatHighBAT}, we extend $M'$ to a model $M^h$ satisfying all axioms of $\highBAT$.

	By the definition of complete abstraction, there is a model $M^l$ of $\lowBAT$ s.t. $M^h \bisimilar_m M^l$.
	We construct $M$ from $M^l$ by considering interpretation on only the symbols in $\lowBAT_{S_0}$.
	Due to the fact that $M^h \bisimilar_m M^l$, we have $S_0^{M^h} \bisimilar_{m}^{M^h, M^l} S_0^{M^l}$.
	It follows that for every high-level predicate fluent $P$, $M^h \models P(S_0)$ iff $M^l \models m(P)[S_0]$, and similarly for every high-level functional fluent $f$.
	Hence, we can extend $M$ to satisfy $\Psi_m$.
	It is easily verified that $M \identical_{\predSet^l \cup \set{\codingSym}} M'$.
		
	($\Leftarrow$):
	Assume that $\highBAT_{S_0} \equiv \forget(\phi \land \lowBAT_{fma} \land \Psi_m, \predSet^l \cup \set{\codingSym})[S_0]$ where $\phi[S_0] = \lowBAT_{S_0}$.
	Let $M^h$ be a model of $\highBAT$.
	It follows that $M^h \models \highBAT_{S_0}$.
	Let $M'$ be a structure s.t. $M'$ and $M^h$ agree on interpretation on $\predSet^h$ and $\funcSet^h$ on initial situation $S_0$.
	Thus, $M' \models \highBAT_{S_0}$.
	There is a model $M$ of $\lowBAT_{S_0} \land \lowBAT_{fma} \land \Psi_m[S_0]$ s.t. $M \identical_{\predSet^l \cup \set{\codingSym}} M'$.
	By Lemma \ref{lem:relSatLowBAT}, $M$ can be extended to a model $M^l$ satisfying all axioms of $\lowBAT$.
	It is easily verified that $S_0^{M^h} \bisimilar_{m}^{M^h, M^l} S_0^{M^l}$.	
	By Lemma \ref{lem:compSoundAbs3}, $M^h \bisimilar_m M^l$.
\end{proof}

\begin{lemma}\label{lem:compSoundAbs}
Let $M^h$ be a model of $\BAT^h$ and $M^l$ a model of $\BAT^l$ s.t. $M^h \bisimilar_m M^l$.
Let $B$ be the $m$-bisimulation between $M^h$ and $M^l$.
Then, the following hold:
\begin{itemize}
\item for every executable situation $s_h$ in $M^h$, there is an $m$-reachable situation $s_l$ in $M^l$ s.t. $(s_h, s_l) \in B$.

\item for every $m$-reachable situation $s_l$ in $M^l$, there is an executable situation $s_h$ in $M^h$ s.t. $(s_h, s_l) \in B$.
\end{itemize}
\end{lemma}
\begin{proof}
	We here only prove the first item.
	The second item can be similarly proved.
	We prove by induction on the length of situation $s_h$.
	
	\textbf{Base case} ($s_h = S_0^{M^h}$): It is obvious since $M^h \bisimilar_m M^l$ implies that $(S_0^{M^h}, S_0^{M^l}) \in B$.
	
	\textbf{Inductive step}: Let $s^*_h$ be an executable situation in $M^h$ s.t. $M^h, v[s^*/s^*_h, s/s_h] \models \Poss(\act, s^*) \land s = do(\act, s^*)$.
	By the inductive hypothesis, there is an $m$-reachable situation $s^*_l$ in $M^l$ s.t. $(s^*_h, s^*_l) \in B$.
	By the back condition of $m$-bisimulation, there is a situation $s_l$ in $M^l$ s.t. $M^l, v[s^*/s^*_l, s/s_l] \models \Do(m(\act), s^*, s)$ and $(s_h, s_l) \in B$.
	This, together with $m$-reachability of $s^*_l$, implies that $s_l$ is $m$-reachable.

\end{proof}

\begin{lemma} \label{lem:compSoundAbs2}
Suppose that $\BAT^h$ is a complete abstraction of $\BAT^l$ relative to $m$.
Then, for every high-level action $\act$ and every high-level situation-suppressed formula $\phi$, $\BAT^l \models \forall s, s'. \Do(\pi^*_m, S_0, s) \land \Do(m(\act), s, s') \implies m(\phi)[s']$ only if $\BAT^h \models \forall s. \Exec(s) \land \Poss(\act, s) \implies \phi[do(\act, s)]$.	
\end{lemma}
\begin{proof}
	Let $M^h$ be an arbitrary model of $\BAT^h$, $s_h$ an arbitrary high-level situation and $\act \in \actSet^h$ s.t. $M^h, v[s/s_h] \models \Exec(s) \land \Poss(\act, s)$.	
	Since $\BAT^h$ is a complete abstraction of $\BAT^l$ relative to $m$, there is a model of $M^l$ of $\lowBAT$ s.t. $M^h \bisimilar_m M^l$.
	Let $B$ be the $m$-bisimulation between $M^h$ and $M^l$.
	By Lemma \ref{lem:compSoundAbs}, there is an $m$-reachable situation $s_l$ in $M^l$ s.t. $(s_h, s_l) \in B$.
	Since $M^h, v[s/s_h] \models \Poss(\act, s)$ and $M^h \bisimilar_m M^l$, there is a situation $s'_l$ in $M^l$ s.t. $M^l \models \Do(m(\act), s_l, s'_l)$ and $(do(\act, s_h), s'_l) \in B$.
	It follows that $do(\act, s_h) \bisimilar^{M^h, M^l}_m s'_l$.
	By Definition \ref{def:isomorphic}, $M^h, v[s/s_h] \models \phi[do(\act, s)]$ and hence $M^h \models \forall s. \phi[do(\act, s)]$.
\end{proof}

{\noindent \bf Proof of Proposition \ref{prop:guardedAct}}
\begin{proof}
We remind that $m(\act) \doteq \pi \vec{x}. \psi(\vec{x})?; A(\vec{x}, \vec{c})$.

\begin{enumerate}
	\item By the definition of $P$-enabling property of $\act$, $A(\vec{x}, \vec{c})$ is alternately $\phi(\vec{y})$-enabling under $\psi(\vec{x})$.
		By the definition of alternative enabling property of $A(\vec{x}, \vec{c})$, $\BAT^l \models \forall s, \vec{x}. (\Poss(A(\vec{x}, \vec{c}), s) \land \psi(\vec{x})[s]) \implies (\neg \phi(\vec{y})[s] \land \phi(\vec{y})[do(A(\vec{x}, \vec{c}), s)])$.
		Hence, $\BAT^l \models \forall s, s', \vec{x}. (\Do(\pi^*_m, S_0, s) \land \psi(\vec{x})[s] \land \Poss(A(\vec{x}, \vec{c}), s) \land s' = do(A(\vec{x}, \vec{c}), s)) \implies \exists \vec{y} \phi(\vec{y})[s']$.
		Then, $\BAT^l \models \forall s, s', \vec{x}. \Do(\pi^*_m, S_0, s) \land \Do(\pi \vec{x}. \psi(\vec{x})?; A(\vec{x}, \vec{c}), s, s') \implies \exists \vec{y} \phi(\vec{y})[s']$.		
		Lemma \ref{lem:compSoundAbs2}, together with the fact that $m(P) \doteq \exists \vec{y}. \phi(\vec{y})$, implies that $\BAT^h \models \forall s. \Exec(s) \land \Poss(\act, s) \implies P(do(\act, s))$.

	\item The proof of this item is similar to Item 1 except we use $P$-disabling property of $\act$ and the exclusion property of $\phi(\vec{y})$.
	
		Since $\phi(\vec{y})$ is exclusive, there is a unique tuple $\vec{y}$ s.t. $\BAT^l \models \forall s. (\Exec(s) \land \Poss(A(\vec{x}, \vec{c}), s) \land \psi(\vec{x})[s]) \implies \phi(\vec{y})[s]$.
		Hence, $\BAT^l \models \forall s, s', \vec{x}. (\Do(\pi^*_m, S_0, s) \land  \psi(\vec{x})[s] \land \Poss(A(\vec{x}, \vec{c}), s) \land s' = do(A(\vec{x}, \vec{c}), s)) \implies \forall \vec{y} (\neg \phi(\vec{y}))[s']$.
		Then, $\BAT^l \models \forall s, s', \vec{x}. \Do(\pi^*_m, S_0, s) \land \Do(\pi \vec{x}. \psi(\vec{x})?; A(\vec{x}, \vec{c}), s, s') \implies \forall \vec{y} \neg \phi(\vec{y})[s']$.
		By Lemma \ref{lem:compSoundAbs2}, $\BAT^h \models \forall s. \Exec(s) \land \Poss(\act, s) \implies \neg P(do(\act, s))$.

	\item The proof of this item is similar to Item 1 except we use $P$-invariant property of $\act$. 

	\item By the definition of $f$-incremental property of $\act$, $A(\vec{x}, \vec{c})$ is singly $\phi(\vec{y})$-enabling under $\psi(\vec{x})$.
		By the definition of single enabling property of $A(\vec{x}, \vec{c})$, $\BAT^l \models \forall s, \vec{x}. (\Poss(A(\vec{x}, \vec{c}), s) \land \psi(\vec{x})[s]) \implies (\neg \phi(\vec{y})[s] \land \phi(\vec{y})[do(A(\vec{x}, \vec{c}), s)])$ and $\lowBAT \models \forall s, \vec{x}. (\Poss(A(\vec{x}, \vec{c}), s) \land \psi(\vec{x})[s]) \implies \forall \vec{z} (\vec{y} \neq \vec{z} \implies \phi(\vec{z}, \vec{d})[s] \equiv \phi(\vec{z}, \vec{d})[do(A(\vec{x}, \vec{c}), s)])$.
		Hence, $\BAT^l \models \forall s, s', \vec{x}. (\Do(\pi^*_m, S_0, s) \land \psi(\vec{x})[s] \land \Poss(A(\vec{x}, \vec{c}), s) \land s' = do(A(\vec{x}, \vec{c}), s)) \implies \countOper \vec{y}. \phi(\vec{y})[s'] = \countOper \vec{y}. \phi(\vec{y})[s] + 1$.
		Then, $\BAT^l \models \forall s, s', \vec{x}. (\Do(\pi^*_m, S_0, s) \land \Do(\pi \vec{x}. \psi(\vec{x})?; A(\vec{x}, \vec{c}), s, s')) \implies \countOper \vec{y}. \phi(\vec{y})[s'] = \countOper \vec{y}. \phi(\vec{y})[s] + 1$.
		Lemma \ref{lem:compSoundAbs2}, together with the fact that $m(f) \doteq \countOper \vec{y}. \phi(\vec{y})$, implies that $\BAT^h \models \Exec(s) \land \Poss(\act, s) \implies f(do(\act, s)) = f(s) + 1$.

	\item The proof of this item is similar to that of the 4th item.

	\item The proof of this item is similar to that of the 4th item.
\end{enumerate}
\end{proof}

{\noindent \bf Proof of Theorem \ref{thm:abstraction}}
\begin{proof}	
	We first prove that $\BAT^h$ is a sound abstraction of $\BAT^l$ relative to $m$.
	Let $M^l$ be a model of $\BAT^l$.
	We construct a model $M^h$ of $\BAT^h$ in the following.
	
	The domain $Act$ of $M^h$ for the sort action has the same size of $\actSet^h$ and each action symbol $\act$ of $\actSet^h$ corresponds to exactly one element $\act^{M^h}$ of $Act$.
	The domain $Sit$ of $M^h$ for the sort situation is the set of all finite sequences of elements of $Act$.
	The interpretation of initial situation symbol $S_0$, function $do$ and relation $\subHis$ are 
	\begin{itemize}
		\item $S_0^{M^h} = []$, the empty sequence;
		
		\item for $\beta^{M^h} \in Act$ and $[\act^{M^h}_1, \cdots, \act^{M^h}_n] \in Sit$,
		\begin{center}		
			$do^{M^h} (\beta^{M^h}, [\act^{M^h}_1 , \cdots, \act^{M^h}_n]) = [\act^{M^h}_1 , \cdots, \act^{M^h}_n , \beta^{M^h}]$;
		\end{center}		 
		
		\item $s  \subHis^{M^h}  s'$ iff the sequence $s$ is an initial subsequence of $s'$.
	\end{itemize}

	Next, we interpret predicate and functional fluents by induction on the situation $s$.
		
	\textbf{Base case} ($s = S^{\highModel}_0$):
	\begin{itemize}
		\item for every high-level predicate fluent $P$, 
			\begin{center}
				$\highModel \models P(S_0)$ iff $\lowModel \models m(P)[S_0]$;
			\end{center}
		
		\item for every high-level functional fluent $f$,
			\begin{center}
				$f^{\highModel}(S_0) = (m(f)[S_0])^{\lowModel}$;
			\end{center}
			
		\item for every high-level action $\act$, 
			\begin{center}
				$M^h \models \Poss(\act, S_0)$ iff \\
				$\highModel \models \forget((\exists \vec{x}. \psi(\vec{x}) \land \Pi_A(\vec{x}, \vec{c})) \land \BAT^{l-}_{con} \land$ \\ 
				\hspace*{20mm}$\lowBAT_{fma} \land \Psi_m, \predSet^l \cup \set{\codingSym})[S_0]$
			\end{center}
			where $m(\act) \doteq \pi \vec{x}. \psi(\vec{x})?; A(\vec{x}, \vec{c})$. 
	\end{itemize}
	
	\textbf{Inductive step}: 
	By induction assumption, $\highModel$ interprets all predicate and functional fluents at situation $s$.
	We interprets them at situation $do(a, s)$.
	\begin{itemize}
		\item for every high-level predicate fluent $P$, every action $\act$ and every situation $s_h$, 
		\begin{itemize}
			\item $\highModel, v[a/\act, s/s_h] \models P(do(a, s))$, if $\act$ is $P$-enabling relative to $m$;
			\item $\highModel, v[a/\act, s/s_h] \nmodels P(do(a, s))$, if $\act$ is $P$-disabling relative to $m$;
			\item $\highModel, v[a/\act, s/s_h] \models P(do(a, s))$ iff $\highModel, v[s/s_h] \models P(s)$, if $\act$ is $P$-invariant relative to $m$;
		\end{itemize}
		
		\item for every high-level functional fluent $f$, every action $\act$, every situation $s_h$,
		\begin{itemize}
			\item $\highModel, v[a/\act, s/s_h] \models f(do(a, s)) = y + 1$ iff $\highModel, v[s/s_h] \models f(s) = y$, if $\act$ is $f$-incremental relative to $m$;
			\item $\highModel, v[a/\act, s/s_h] \models f(do(a, s)) = y - 1$ iff $\highModel, v[s/s_h] \models f(s) = y$, if $\act$ is $f$-decremental relative to $m$;
			\item $\highModel, v[a/\act, s/s_h] \models f(do(a, s)) = y$ iff $\highModel, v[s/s_h] \models f(s) = y$, if $\act$ is $f$-invariant relative to $m$;
		\end{itemize}
	\end{itemize}
	
	Since $m$ is complete (that is, $\act$ is $P$-enabling, $P$-disabling, or $P$-invariant; and $\act$ is $f$-incremental, $f$-decremental, or $f$-invariant), the interpretation of every high-level predicate fluent $P$ and functional fluent $f$ at situation $do(a, s)$ is well-defined.
	
	Finally, we interpret the predicate $\Poss$ as:	
	
	For every high-level action $\act$ and high-level situation $s_l$, 
	\begin{center}
		$M^h, v[s/s_l] \models \Poss(\act, s)$ iff \\
		$\highModel, v[s/s_l] \models \forget((\exists \vec{x}. \psi(\vec{x}) \land \Pi_A(\vec{x}, \vec{c})) \land \BAT^{l-}_{con} \land$ \\ 
		\hspace*{27.5mm} $\lowBAT_{fma} \land \Psi_m, \predSet^l \cup \set{\codingSym})[s]$
	\end{center}
	where $m(\act) \doteq \pi \vec{x}. \psi(\vec{x})?; A(\vec{x}, \vec{c})$. 
	
	We hereafter show that $\highModel \models \highBAT$.	
	It is easily verified that the model $M^h$ satisfies the foundational axiom $\Sigma$, the unique names axioms for actions $\highBAT_{una}$ and the action precondition axiom $\highBAT_{ap}$.
		
	We now show that $\highModel \models \highBAT_{S_0}$.
	We can construct a model $M$ by combining the interpretations of $\lowModel$ and $\highModel$ at $S_0$, that is, 
	\begin{itemize}		
		\item $M$ and $\lowModel$ agree on the interpretation of every low-level predicate fluent at $S_0$ and coding symbol $\codingSym$;
		
		\item $M$ and $\highModel$ agree on the interpretation of every high-level predicate and functional fluent at $S_0$.
	\end{itemize}	
	By the above interpretation of $M^h$ at $S_0$ and the fact that $\lowModel \models \lowBAT_{S_0} \land \lowBAT_{fma}$, we get that $M \models \lowBAT_{S_0} \land \lowBAT_{fma} \land \Psi_m[S_0]$.
	It follows from the definition of forgetting (cf. Definition \ref{def:forget}) that $\highModel \models \highBAT_{S_0}$.
		
	We then show that $\highModel \models \highBAT_{ss}$.
	We first verify that for every high-level predicate fluent $P$, every action $\act$ and every situation $s_h$, $\highModel, v[a/\act, s/s_h] \models$
	\begin{center}
		$P(do(a, s)) \equiv [\bigvee \limits_{i = 1}^{m} a = \act^+_i] \lor [P(s) \land \bigwedge \limits_{j = 1}^{n} a \neq \act^-_j]$
	\end{center}
	where each $\act^+_i$ is $P$-enabling relative to $m$ and each $\act^-_j$ is $P$-disabling.
	
	Suppose that $\highModel, v[a/\act, s/s_h] \models P(do(a, s))$.
	By the interpretation on $P(do(a, s))$ and the completeness of the refinement mapping $m$, only two cases hold: (1) $\act$ is $P$-enabling relative to $m$; and (2) $\act$ is $P$-invariant and $\highModel, v[s/s_h] \models P(s)$.
	These two cases implies that $\highModel, v[a/\act, s/s_h] \models [\bigvee \limits_{i = 1}^{m} a = \act^+_i] \lor [P(s) \land \bigwedge \limits_{j = 1}^{n} a \neq \act^-_j]$.	
	Suppose that $\highModel, v[a/\act, s/s_h] \nmodels P(do(a, s))$.
	Similarly to the above proof, we can get that $\highModel, v[a/\act, s/s_h] \models [\bigvee \limits_{i = 1}^{m} a = \act^+_i] \lor [P(s) \land \bigwedge \limits_{j = 1}^{n} a \neq \act^-_j]$.
	
	The situation for high-level functional fluent $f$ can be similarly proved via the interpretation on $f(do(a, s))$ and the completeness of $m$.
		
	We then verify that $M^h$ is $m$-bisimilar to $M^l$.
	We construct the relation $B$ over $\sitDom^{M^h} \times \sitDom^{M^l}$ as follows: $(s_h, s_l) \in B$ iff there is a high-level action sequence $\vec{\act}$ s.t. $M^l, v[s/s_l] \models \Do(m(\vec{\act}), S_0, s)$ and $s_h = do(\vec{\act}, S_0)^{M^h}$.
	Clearly, $(S_0^{M^h}, S_0^{M^l}) \in B$.
	It is easily verified that every low-level situation $s_l$ s.t. $(s_h, s_l) \in B$ is $m$-reachable.
	We show that every high-level situation $s_h$ s.t. $(s_h, s_l) \in B$ is executable and that $B$ is an $m$-bisimulation relation between $M^h$ and $M^l$ by induction on the length of situation $s_h$. 
	
	\textbf{Base case} ($s_h = S_0^{M^h}$): Clearly, $S_0^{M^h}$ is executable.
	In the following, we verify the atom, forth and back conditions of $m$-bisimulation.
	\begin{description}
		\item[atom] $S_0^{M^h}$ and $S_0^{M^l}$ are $m$-isomorphic due to the construction of $M^h$.
		
		\item[forth] Let $\beta$ be a high-level action where $m(\beta) \doteq \pi \vec{x}. \psi(\vec{x})?; A(\vec{x}, \vec{c})$. 
		Let $s'_l$ be a low-level situation s.t. $M^l, v[s'/s'_l] \models \Do(m(\beta), S_0, s')$.
		We let $s'_h$ be $do(\beta, S_0)^{M^h}$.
		By the construction of relation $B$, $(s'_h, s'_l) \in B$.
		It remains to show that $M^h \models \Poss(\beta, S_0)$.
		Since $M^l, v[s'/s'_l] \models \Do(m(\beta), S_0, s')$, $M^l \models (\exists \vec{x}. \psi(\vec{x}) \land \Pi_A(\vec{x}, \vec{c}))[S_0]$.	
		In addition, $M^l \models \BAT^{l-}_{con}[S_0] \land \lowBAT_{fma}$.		
		By the above interpretation of $M^h$ at $S_0$ and the construction of $M$, $M \models (\exists \vec{x}. \psi(\vec{x}) \land \Pi_A(\vec{x}, \vec{c}))[S_0] \land  \BAT^{l-}_{con}[S_0] \land \lowBAT_{fma} \land \Psi_m[S_0]$.
		By the definition of forgetting (cf. Definition \ref{def:forget}), we get that $M^h \models \forget((\exists \vec{x}. \psi(\vec{x}) \land \Pi_A(\vec{x}, \vec{c})) \land \BAT^{l-}_{con} \land \lowBAT_{fma} \land \Psi_m, \predSet^l \cup \set{\codingSym})[S_0]$.
		Hence, $M^h \models \Poss(\beta, S_0)$.
		
		\item[back] Let $\beta$ be a high-level action where $m(\beta) \doteq \pi \vec{x}. \psi(\vec{x})?; A(\vec{x}, \vec{c})$. 
		Let $s'_h$ be a high-level situation s.t. $M^h, v[s'/s'_h] \models \Poss(\beta, S_0) \land s' = do(\beta, S_0)$.
		It follows that $M^h \models \forget((\exists \vec{x}. \psi(\vec{x}) \land \Pi_A(\vec{x}, \vec{c})) \land \BAT^{l-}_{con} \land \lowBAT_{fma} \land \Psi_m, \predSet^l \cup \set{\codingSym})[S_0]$.
		There is a model $M_*$ s.t. $M_* \models (\exists \vec{x}. \psi(\vec{x}) \land \Pi_A(\vec{x}, \vec{c}))[S_0] \land \BAT^{l-}_{con}[S_0] \land \lowBAT_{fma} \land \Psi_m$ and $M_*$ and $M^h$ agree on high-level predicate and functional fluents at $S_0$.
		We construct $M^l_*$ s.t. $M^l_*$ and $M_*$ agree on low-level predicate fluents at $S_0$.
		Clearly, $M^l_* \models \lowBAT_{fma}$.
		By Lemma \ref{lem:relSatLowBAT}, $M^l_*$ can be extended to satisfy all axioms of $\lowBAT$.
		It is easily verified that $M^l_* \models \exists s. \Do(m(\beta), S_0, s)$ and $S^{M^h}_0 \sim_m S^{M^l_*}_0$.
		By construction of $M^h$, $S^{M^h}_0 \sim_m S^{M^l}_0$.
		By the executability-preserving property of refinement mapping $m$, $M^l \models \exists s. \Do(m(\beta), S_0, s)$.
		Hence, there is a low-level situation $s'_l$ s.t. $M^l, v[s'/s'_l] \models \Do(m(\beta), S_0, s'_l)$.
		By the construction of relation $B$, $(s'_h, s'_l) \in B$.
	\end{description}
	
	\textbf{Inductive step}: Suppose that $(s'_h, s'_l) \in B$ where $s_h = do(\vec{\act}, S_0)^{M^h}$ and $s'_h = do(\beta, s_h)^{M^h}$ where $m(\beta) \doteq \pi \vec{x}. \psi(\vec{x})?; A(\vec{x}, \vec{c})$. 
	We here only verify that $s'_h$ is executable and that $s'_h$ and $s'_l$ are $m$-isomorphic.
	Similarly to the base case, the forth and back conditions of $m$-bisimulation can be proved.
	By the construction of relation $B$ and the semantics of Golog programs, there is a situations $s_l$ in $M^l$ s.t. $M^l, v[s/s_l, s'/s'_l] \models \Do(m(\vec{\act}), S_0, s) \land \Do(m(\beta), s, s')$. 
	It follows that there is a tuple $\vec{e}$ of objects s.t. $M^l, v[s/s_l, s'/s'_l, \vec{x}/\vec{e}] \models \psi(\vec{x})[s] \land \Pi_A(\vec{x}, \vec{c})[s] \land s' = do(A(\vec{x}, \vec{c}), s)$.	
	Hence, $M^l, v[s/s_l] \models \exists (\vec{x}. \psi(\vec{x}) \land \Pi_A(\vec{x}, \vec{c}))[s]$.
	In addition, $s_l$ is executable and hence $M^l, v[s/s_l] \models \BAT^{l-}_{con}[s]$.
	Furthermore, $M^l, v[s/s_l] \models \lowBAT_{fma}$.
	It is easily verified that $(s_h, s_l) \in B$.
	By the induction hypothesis, $s_h$ is executable and $m$-isomorphic to $s_l$.
	In addition, $M^l, v[s/s_l] \models \exists (\vec{x}. \psi(\vec{x}) \land \Pi_A(\vec{x}, \vec{c}))[s] \land \BAT^{l-}_{con}[s] \land \lowBAT_{fma}$.
	By Proposition \ref{prop:isomorphicForget}, we get that $M^h, v[s/s_h] \models \forget((\exists \vec{x}. \psi(\vec{x}) \land \Pi_A(\vec{x}, \vec{c})) \land \BAT^{l-}_{con} \land \lowBAT_{fma} \land \Psi_m, \predSet^l \cup \set{\codingSym})[s]$.
	Hence, $M^h, v[s/s_h] \models \Poss(\beta, s)$ and $s'_h$ is executable.

	We first verify that for every high-level predicate fluent $P$ and every variable assignment $v$, $M^h, v[s'/s'_h] \models P(s')$ iff $M^l, v[s'/s'_l] \models m(P)[s']$.
	Let $P$ be a high-level predicate fluent where $m(P) \doteq \exists \vec{y}. \phi(\vec{y})$.
	
	\begin{itemize}
		\item ($\Rightarrow$): 
		Suppose that $M^h, v[s'/s'_h] \models P(s')$.
		Since the mapping $m$ is complete, $\beta$ is either $P$-enabling or $P$-invariant relative to $m$.
		
		In the case where $\beta$ is $P$-enabling.
		It follows from Definition \ref{def:highActEnabling} that $A(\vec{x}, \vec{c})$ is alternately $\phi(\vec{y})$-enabling under $\psi(\vec{x})$.
		This together with the fact that $M^l, v[s/s_l, \vec{x}/\vec{e}] \models \Poss(A(\vec{x}, \vec{c}), s) \land \psi(\vec{x})[s]$ implies that $M^l, v[s/s_l, \vec{x}/\vec{e}]  \models \phi(\vec{y})[do(A(\vec{x}, \vec{c}), s)]$ where $\vec{x}$ includes all of the variables of $\vec{y}$.
		Hence, $M^l, v[s'/s'_l] \models m(P)[s']$.
		
		In the case where $\beta$ is $P$-invariant.
		According to successor state axiom for $P$, we get that $M^h, v[s/s_h, s'/s'_h] \models \Exec(s) \land \Poss(\beta, s) \implies P(s) \equiv P(s')$.
		It follows that $M^h, v[s/s_h] \models P(s)$.
		By Proposition \ref{prop:isomorphic}, $M^l, v[s/s_l] \models m(P)[s]$.
		Hence, $M^l, v[s/s_l] \models \exists \vec{y}. \phi(\vec{y})[s]$.
		The $P$-invariant property of $\beta$, together with the fact that $M^l, v[s/s_l, \vec{x}/\vec{e}] \models \Poss(A(\vec{x}, \vec{c}), s) \land \psi(\vec{x})[s]$, implies that $M^l, v[s/s_l, \vec{x}/\vec{e}] \models \forall \vec{y}. \phi(\vec{y})[s] \equiv \phi(\vec{y})[do(A(\vec{x}, \vec{c}), s)]$.
		It follows that $M^l, [s'/s'_l] \models  m(P)[s']$. 
		
		\item ($\Leftarrow$): We prove by contraposition.
		Suppose that $M^h, v[s'/s'_h] \models \neg P(s')$.
		Since the mapping $m$ is complete, $\beta$ is either $P$-disabling or $P$-invariant relative to $m$.
		
		In the case where $\beta$ is $P$-disabling.
		It follows from Definition \ref{def:highActEnabling} that $A(\vec{x}, \vec{c})$ is alternately $\phi(\vec{y})$-disabling under $\psi(\vec{x})$.
		This together with the fact that $M^l, v[s/s_l, \vec{x}/\vec{e}] \models \Poss(A(\vec{x}, \vec{c}), s) \land \psi(\vec{x})[s]$ implies that $M^l, v[s/s_l, \vec{x}/\vec{e}] \models \neg \phi(\vec{y})[do(A(\vec{x}, \vec{c}), s)]$ where $\vec{x}$ includes all of the variables of $\vec{y}$.
		In addition, $\phi(\vec{y})$ is exclusive.
		Hence, $M^l, v[s'/s'_l] \models \neg \exists \vec{y}. \phi(\vec{y})[s']$ and $M^l, v[s'/s'_l] \models m(\neg P)[s']$.
		
		The proof for the case where $\beta$ is $P$-invariant is similar the only-if direction.
	\end{itemize}
	
	We now verify that for every high-level functional fluent $f$ and every variable assignment $v$, $M^h, v[s'/s'_h] \models f(s') = z$ iff $M^l, v[s'/s'_l] \models m(f)[s'] = z$.
	Let $f$ be a high-level functional fluent where $m(f) \doteq \countOper \vec{y}. \phi(\vec{y})$.
	\begin{itemize}
		\item ($\Rightarrow$): Suppose that $M^h, v[s'/s'_h] \models f(s') = c$.
		Since the mapping $m$ is complete, $\beta$ is $f$-incremental, $f$-decremental or $f$-invariant relative to $m$.
		
		In the case where $\beta$ is $f$-incremental.
		$M^h, v[s/s_h] \models f(s) = z - 1$.
		By Proposition \ref{prop:isomorphic}, $M^l, v[s/s_h] \models m(f)[s] = z - 1$.	
		It follows from Definition \ref{def:highActEnabling} that $A(\vec{x}, \vec{c})$ is singly $\phi(\vec{y})$-enabling under $\psi(\vec{x})$.
		This together with the fact that $M^l, v[s/s_l, \vec{x}/\vec{e}] \models \Poss(A(\vec{x}, \vec{c}), s) \land \psi(\vec{x})[s]$ implies that $M^l, v[s/s_l, \vec{x}/\vec{e}]  \models \countOper \vec{y}. \phi(\vec{y})[do(A(\vec{x}, \vec{c}), s)] = \countOper \vec{y}. \psi(\vec{y})[s] + 1$ where $\vec{x}$ includes all of the variables of $\vec{y}$.
		Hence, $M^l, v[s/s_l, \vec{x}/\vec{e}] \models m(f)[do(A(\vec{x}, \vec{c}), s)] = m(f)[s] + 1$ and $M^l, v[s'/s'_l] \models m(f)[s'] = z$.
		
		In the case where $\beta$ is $f$-invariant.
		The successor state axiom for $f$, together with the fact that $M^h \models v[s/s_h] \models \Exec(s) \land \Poss(\beta, s)$, implies that $M^h, v[s/s_h, s'/s'_h] \models f(s) = f(s')$.
		It follows that $M^h, v[s/s_h] \models f(s) = z$.
		By Proposition \ref{prop:isomorphic}, $M^l, v[s/s_l] \models \countOper \vec{y}. \phi(\vec{y})[s] = z$.
		The $f$-invariant property of $\beta$, together with the fact that $M^l, v[s/s_l, \vec{x}/\vec{e}] \models \Poss(A(\vec{x}, \vec{c}), s) \land \psi(\vec{x})[s]$, implies that $M^l, v[s/s_l, \vec{x}/\vec{e}] \models \forall{y}. \phi(\vec{y})[s] \equiv \phi(\vec{y})[do(A(\vec{x}, \vec{c}), s)]$.
		It follows that $M^l, v[s/s_l, \vec{x}/\vec{e}] \models \countOper \vec{y}. \phi(\vec{y})[do(A(\vec{x}, \vec{c}), s)] = z$ and hence $M^l, v[s'/s'_l] \models m(f)[s'] = z$.
		
		The proof for the case where $\beta$ is $f$-decremental is similar to that for the case where $\beta$ is $f$-incremental.
		
		\item ($\Leftarrow$): Suppose that $M^l, v[s'/s'_l] \models m(f)[s'] = z$, that is $M^l, v[s'/s'_l] \models \countOper \vec{y}. \phi(\vec{y})[s'] = z$.
		Since the mapping $m$ is complete, $\beta$ is $f$-incremental, $f$-decremental or $f$-invariant relative to $m$.
		
		In the case where $\beta$ is $f$-incremental.
		It follows from Definition \ref{def:highActEnabling} that $A(\vec{x}, \vec{c})$ is singly $\phi(\vec{y})$-enabling under $\psi(\vec{x})$.
		This together with the fact that $M^l, v[s/s_l, \vec{x}/\vec{e}] \models \Poss(A(\vec{x}, \vec{c}), s) \land \psi(\vec{x})(s)$ implies that $M^l, v[s/s_l, \vec{x}/\vec{e}]  \models \countOper \vec{y}. \phi(\vec{y})[do(A(\vec{x}, \vec{c}), s)] = \countOper \vec{y}. \phi(\vec{y})[s] + 1$ where $\vec{x}$ includes all of the variables of $\vec{y}$.
		Hence, $M^l, v[s/s_l, \vec{x}/\vec{e}] \models m(f)[do(A(\vec{x}, \vec{c}), s)] = m(f)[s] + 1$ and $M^l, v[s/s_l] \models m(f)[s] = z - 1$.
		In addition, $M^h, v[s/s_h] \models \Exec(s) \land \Poss(\beta, s)$.
		By Proposition \ref{prop:isomorphic}, $M^h, v[s/s_h] \models f[s] = z - 1$.
		The successor state axiom for $f$ together with the $f$-incremental property of $\beta$ implies that $M^h, v[s'/s'_h] \models f[s'] = z$.
		
		In the case where $\beta$ is $f$-invariant.
		It follows from Definition \ref{def:highActEnabling} that $A(\vec{x}, \vec{c})$ is $\phi(\vec{y})$-invariant under $\psi(\vec{x})$.
		This together with the fact that $M^l, v[s/s_l, \vec{x}/\vec{e}] \models \Poss(A(\vec{x}, \vec{c}), s) \land \psi(\vec{x})(s)$ implies that $M^l, v[s/s_l, \vec{x}/\vec{e}] \models \countOper \vec{y}. \phi(\vec{y})[do(A(\vec{x}, \vec{c}), s)] = \countOper \vec{y}. \phi(\vec{y})[s]$ where $\vec{x}$ includes all of the variables of $\vec{y}$.
		Hence, $M^l, v[s/s_l, \vec{x}/\vec{e}] \models m(f)[do(A(\vec{x}, \vec{c}), s)] = m(f)[s]$ and $M^l, v[s/s_l] \models m(f)[s] = z$.	
		By Proposition \ref{prop:isomorphic}, $M^h, v[s/s_h] \models f[s] = z$.
		In addition, $M^h, v[s/s_h] \models \Exec(s) \land \Poss(\beta, s)$.
		The successor state axiom for $f$ together with the $f$-incremental property of $\beta$ implies that $M^h, v[s'/s'_h] \models f[s'] = z$.
		
		The proof for the case where $\beta$ is $f$-decremental is similar to that for the case where $\beta$ is $f$-incremental.
	\end{itemize}
	
	Finally, since $\BAT^h_{S_0} \equiv \forget(\phi \land \lowBAT_{fma} \land \Psi_m, \predSet^l \cup \set{\codingSym})[S_0]$ where $\phi[S_0] = \lowBAT_{S_0}$, by Theorem \ref{thm:compSoundAbs}, $\BAT^h$ is a complete abstraction of $\BAT^l$ relative to $m$.
\end{proof}

{\noindent \bf Proof of Proposition \ref{prop:BWSimpForg}}
\begin{proof}
	It suffices to prove that $\forget(\phi \land \BAT^{BW_l-}_{con} \land \BAT^{BW_l}_{fma} \land \Psi_m, \predSet^l \cup \set{\codingSym}) \equiv \forget(\phi \land \Psi_m, \predSet^l)$.
	
	We first prove the only-if direction of the definition of forgetting (cf. Definition \ref{def:forget}).
	Let $M$ be a model of $\forget(\phi \land \BAT^{BW_l-}_{con} \land \BAT^{BW_l}_{fma} \land \Psi_m, \predSet^l \cup \set{\codingSym})$.
	There is a model $M'$ s.t. $M' \models \phi \land \BAT^{BW_l-}_{con} \land \BAT^{BW_l}_{fma} \land \Psi_m$ and $M \identical_{\predSet^l \cup \set{\codingSym}} M'$.
	Clearly, $M' \models \phi \land \Psi_m$ and $M \identical_{\predSet^l} M'$.
	Hence, $M \models \forget(\phi \land \Psi_m, \predSet^l)$.
	
	We now prove that the if direction.
	Let $M$ be a model of $\forget(\phi \land \Psi_m, \predSet^l)$.
	We can construct a model $M'$ as follows.
	The domain of $M'$ contains $Num^{M'} + 2$ blocks.
	We use one block for the interpretation of symbol $C$ and $\codingSym(C) = 1$.
	We build up a tower whose bottom block is $C$ by the other $Num^{M'}$ blocks.
	We let $\codingSym(x)$ be the distance between block $x$ and $C$ plus $1$ for every block $x$ on this tower.	
	If $M \models Holding$, then we let the last block holding in $M'$; otherwise, the last block is not holding.
	We let $\codingSym(y) = 0$ for the last block $y$.			
	It is easily verified that $M' \models \phi \land \BAT^{BW_l-}_{con} \land \BAT^{BW_l}_{fma} \land \Psi_m$ and $M \identical_{\predSet^l \cup \set{\codingSym}} M'$.	
\end{proof}

Let $m$ be a flat refinement mapping and $\varphi$ a propositional $\exists$-formula over $\Phi(m)$.
We use $t(\varphi)$ to denote the formula obtained from $\varphi$ via replacing every occurrence of $\exists \vec{x}. \phi(\vec{x})$ in $\varphi$ by $P$ (resp. $f > 0$) where $m(P) \doteq \exists \vec{x}. \phi(\vec{x})$ (resp. $m(f) \doteq \countOper \vec{x}. \phi(\vec{x})$).

\begin{lemma} \label{lem:refMap}
	Let $m$ be a flat, consistent and syntax-irrelevant refinement mapping and $\varphi$ a propositional $\exists$-formula over $\Phi(m)$.
	Then, $\varphi \land \Psi_m \equiv t(\varphi) \land \Psi_m$.
\end{lemma}
\begin{proof}
	We prove by induction on the structure of $\varphi$.
	
	\begin{itemize}
		\item $\varphi = \exists \vec{x}. \phi(\vec{x})$ and $m(P) \doteq \exists \vec{x}. \phi(\vec{x})$: It follows that $t(\varphi) = P$.
		Since $\Psi_m \models P \equiv \exists \vec{x}. \phi(\vec{x})$, $\exists \vec{x}. \phi(\vec{x}) \land \Psi_m \equiv P \land \Psi_m$.
		
		\item $\varphi = \exists \vec{x}. \phi(\vec{x})$ and $m(f) \doteq \countOper \vec{x}. \phi(\vec{x})$: It follows that $t(\varphi) = f > 0$.
		Since $\Psi_m \models f = \countOper \vec{x}. \phi(\vec{x})$, $\exists \vec{x}. \phi(\vec{x}) \land \Psi_m \equiv (\countOper \vec{x}. \phi(\vec{x}) > 0) \land \Psi_m$.
		Hence, $\exists \vec{x}. \phi(\vec{x}) \land \Psi_m \equiv f > 0 \land \Psi_m$.
		
		\item $\varphi = \neg \phi$: It follows that $t(\varphi) = \neg t(\phi)$.
		By the inductive hypothesis, $\phi \land \Psi_m \equiv t(\phi) \land \Psi_m$.
		We prove $\varphi \land \Psi_m \models t(\varphi) \land \Psi_m$.
		Let $M$ be a model of $\varphi \land \Psi_m$.
		Since $\Psi_m \doteq [\bigwedge_{P \in \predSet^h} P \equiv m(P) ] \land [\bigwedge_{f \in \funcSet^h} f = m(f)]$, $M \nmodels \phi \land \Psi_m$.
		Since $\phi \land \Psi_m \equiv t(\phi) \land \Psi_m$, $M \nmodels t(\phi) \land \Psi_m$.
		Hence, $M \models \neg t(\phi) \land \Psi_m$.
		The opposite direction that $t(\varphi) \land \Psi_m \models \varphi \land \Psi_m$ can be similarly proved.		
		
		\item $\varphi = \phi_1 \land \phi_2$: It follows that $t(\varphi) = t(\phi_1) \land t(\phi_2)$.	
		By the inductive hypothesis, $\phi_1 \land \Psi_m \equiv t(\phi_1) \land \Psi_m$ and $\phi_2 \land \Psi_m \equiv t(\phi_2) \land \Psi_m$.
		Hence, $\phi_1 \land \phi_2 \land \Psi_m \equiv t(\phi_1) \land t(\phi_2) \land \Psi_m$.
	\end{itemize}
\end{proof}

{\noindent \bf Proof of Proposition \ref{prop:consSynIrRefMap}}
\begin{proof}
	We first prove the only-if direction of the definition of forgetting (cf. Definition \ref{def:forget}).
	Let $M$ and $M'$ be two models s.t. $M \models \varphi \land \Psi_m$ and $M \identical_{\predSet^l} M'$.
	Clearly, the formula $m^{-1}(\varphi)$ is $t(\varphi) \land \bigwedge_{f \in \funcSet^h(t(\varphi))} f \geq 0$.
	By Lemma \ref{lem:refMap}, we get that $M \models t(\varphi) \land \Psi_m$.
	In addition, $\Psi_m \models \bigwedge_{f \in \funcSet^h} f = \countOper \vec{x}. \phi_f(\vec{x})$ where $m(f) \doteq \countOper \vec{x}. \phi_f(\vec{x})$, $M \models \bigwedge_{f \in \funcSet^h} f \geq 0$ and hence $M \models m^{-1}(\varphi)$.
	Furthermore, $M$ and $M'$ agree on everything except possibly on the interpretation of every symbol of $\predSet^l$.
	Hence, $M' \models m^{-1}(\varphi)$.
		
	We now prove that the if direction.
	Let $M'$ be a model of $m^{-1}(\varphi)$.
	Since $m$ is syntax-irrelevant and consistent, we can construct a model $M$ as follows.
	For every high-level predicate fluent $P$ s.t. $m(P) \doteq \exists \vec{x}. \phi(\vec{x})$: if $M' \models P$, we let $M \models P \land \exists \vec{x}. \phi(\vec{x})$; otherwise, we let $M \models \neg P \land \forall \vec{x}. \neg \phi(\vec{x})$.
	For every high-level functional fluent $f$ s.t. $m(f) \doteq \countOper \vec{x}. \phi(\vec{x})$: if $M' \models f > 0$, we let $M \models f > 0 \land \exists \vec{x}. \phi(\vec{x})$; otherwise, we let $M \models f = 0 \land \forall \vec{x}. \neg \phi(\vec{x})$.
	It is easily verified that $M \models m^{-1}(\varphi) \land \Psi_m$ and $M \identical_{\predSet^l} M'$.	
\end{proof}

\looseness=-1

\begin{lemma} \label{lem:absState}
	Suppose that two low-level situations $s_1$ and $s_2$ of two low-level models $M_1$ and $M_2$ are in the same abstract state.
	Then, for every high-level situation suppressed formula $\phi$, $M_1, v[s/s_1] \models m(\phi)[s]$ iff $M_{2}, v[s/s_{2}] \models m(\phi)[s]$.
\end{lemma}
\begin{proof}
	By induction on the structure of $\phi$.
\end{proof}

{\noindent \bf Proof of Proposition \ref{prop:execPreRefMap}}
\begin{proof}
	Let $m$ be a flat, consistent, syntax-irrelevant and propositional $\exists$-definable refinement mapping.
	Let $M_1$ and $M_2$ be two models of $\lowBAT$.
	Let $s_1 \in \Delta_S^{M_1}$ and $s_2 \in \Delta_S^{M_2}$ be two low-level $m$-reachable situations in the same abstract state.
	Let $\act$ be a high-level action where $m(\act) \doteq  \pi \vec{x}. \psi(\vec{x})?; A(\vec{x}, \vec{c})$. 

	We here only prove that the only-if direction.
	The if direction can be similarly proved.
	Suppose that $M_1, v[s/s_1] \models \exists s'. \Do(m(\act), s, s')$.
	It follows that $M_1, v[s/s_1] \models (\exists \vec{x}. \psi(\vec{x}) \land \Pi_{A}(\vec{x}, \vec{c}))[s]$.
	By the propositional $\exists$-definability of refinement mapping $m$, there is a propositional $\exists$-formula $\phi$ s.t. $\BAT^{l-}_{con} \land \lowBAT_{fma} \models \phi \equiv \exists \vec{x}. \psi(\vec{x}) \land \Pi_A(\vec{x}, \vec{c})$.
	Let $\varphi = m^{-1}(\phi)$ is a high-level situation-suppressed formula.
	
	Since $m$ and $m^{-1}$ in fact are two reverse translations, $m(\varphi) \equiv \phi$.
	Thus, $M_1, v[s/s_1] \models m(\varphi)[s]$.
	By Lemma \ref{lem:absState}, $M_2, v[s/s_2] \models m(\varphi)[s]$.
	It follows that $M_2, v[s/s_2] \models \phi[s]$ and $M_2, v[s/s_2] \models (\exists \vec{x}. \psi(\vec{x}) \land \Pi_A(\vec{x}, \vec{c}))[s]$.
	Hence, $M_2, v[s/s_2] \models \exists s'. \Do(\act, s, s')$.
\end{proof}

\end{document}